\documentclass[11pt,letterpaper]{article}

\usepackage[margin=1in]{geometry}

\usepackage{theorem,latexsym,graphicx,amssymb}
\usepackage{amsmath,enumerate}
\usepackage{float}
\usepackage{xspace}
\usepackage{paralist}
\usepackage{enumerate}
\usepackage{cases}
\usepackage{caption}
\usepackage{natbib}

\newenvironment{proof}{{\bf Proof:  }}{\hfill\rule{2mm}{2mm}\vspace*{5pt}}
\newenvironment{proofof}[1]{{\vspace*{5pt} \noindent\bf Proof of #1:  }}{\hfill\rule{2mm}{2mm}\vspace*{5pt}}
\newtheorem{definition}{Definition}[section]

\newtheorem{theorem}{Theorem}[section]
\newtheorem{lemma}{Lemma}[section]
\newtheorem{claim}{Claim}[section]
\newtheorem{fact}{Fact}[section]
\newtheorem{corollary}{Corollary}[section]
\newtheorem{observation}{Observation}[section]

\newcommand{\omm}{\textsf{Online Makespan Minimization}\xspace}
\newcommand{\lptr}{\textsf{LPT with Restart}\xspace}
\newcommand{\alg}{\textsf{ALG}}
\newcommand{\opt}{\textsf{OPT}}
\newcommand{\I}{\mathsf{I}}
\newcommand{\Q}{\mathsf{Q}}
\newcommand{\W}{\mathsf{W}}
\newcommand{\A}{\mathsf{A}}
\renewcommand{\P}{\mathsf{P}}

\title{Online Makespan Minimization: The Power of Restart\footnote{A preliminary version of the paper to appear in APPROX 2018.}}

\author{Zhiyi Huang\thanks{Department of Computer Science, The University of Hong Kong. \texttt{zhiyi@cs.hku.hk}. Partially supported by the Hong Kong RGC under the grant HKU17202115E.}
	\and Ning Kang\thanks{Department of Computer Science, The University of Hong Kong. \texttt{nkang@cs.hku.hk}.}
	\and Zhihao Gavin Tang\thanks{Department of Computer Science, The University of Hong Kong. \texttt{zhtang@cs.hku.hk}.}
	\and Xiaowei Wu\thanks{Department of Computing, The Hong Kong Polytechnic University. \texttt{wxw0711@gmail.com.} Part of the work was done when the author was a postdoc at the University of Hong Kong.}
	\and Yuhao Zhang\thanks{Department of Computer Science, The University of Hong Kong. \texttt{yhzhang2@cs.hku.hk}.}
}
\date{}

\begin{document}

\begin{titlepage}
	\thispagestyle{empty}
	\maketitle
	
	\begin{abstract}
		We consider the online makespan minimization problem on identical machines. Chen and Vestjens (ORL 1997) show that the largest processing time first (LPT) algorithm is 1.5-competitive. For the special case of two machines, Noga and Seiden (TCS 2001) introduce the SLEEPY algorithm that achieves a competitive ratio of $(5 - \sqrt{5})/2 \approx 1.382$, matching the lower bound by Chen and Vestjens (ORL 1997). Furthermore, Noga and Seiden note that in many applications one can kill a job and restart it later, and they leave an open problem whether algorithms with restart can obtain better competitive ratios.
		
		We resolve this long-standing open problem on the positive end. Our algorithm has a natural rule for killing a processing job: a newly-arrived job replaces the smallest processing job if 1) the new job is larger than other pending jobs, 2) the new job is much larger than the processing one, and 3) the processed portion is small relative to the size of the new job. With appropriate choice of parameters, we show that our algorithm improves the 1.5 competitive ratio for the general case, and the 1.382 competitive ratio for the two-machine case.
	\end{abstract}
\end{titlepage}

\section{Introduction}\label{sec:intro}

We study in this paper the classic online scheduling problem on identical machines.
Let there be $m$ identical machines, and a set of jobs that arrive over time.
For each job $j$, let $r_j$ denote its release time (arrival time), and $p_j$ denote its processing time (size).
We assume without loss of generality that all $r_j$'s and $p_j$'s are distinct.
We seek to schedule each job on one of the $m$ machines such that the makespan (the completion time of the job that completes last) is minimized.

We adopt the standard assumption that there is a pending pool such that jobs released but not scheduled are in the pending pool.
That is, the algorithm does not need to assign a job to one of the machines at its arrival; it can decide later when a machine becomes idle.
Alternatively, the \emph{immediate-dispatching} model has also been considered in some papers (e.g.,~\cite{algorithmica/AvrahamiA07}).

We consider the standard competitive analysis of online algorithms.
An algorithm is $(1+\gamma)$-competitive if for any online sequence of jobs, the makespan of the schedule made by the algorithm is at most $(1+\gamma)$ times the minimum makespan in hindsight.
Without loss of generality, (by scaling the job sizes) we assume the minimum makespan $\opt =1$ (for analysis purpose only).

\citet{orl/ChenV97} consider a greedy algorithm called \emph{largest processing time first} (LPT): whenever there is an idle machine, schedule the largest job in the pending pool.
They prove that the LPT algorithm is $1.5$-competitive and provide a matching lower bound (consider $m$ jobs of size $0.5$ followed by a job of size $1$).
They also show that no online algorithm can achieve a competitive ratio better than $1.3473$. 
For the special case when there are only two machines, \citet{tcs/NogaS01} introduce the \textsf{SLEEPY} algorithm that achieves a tight $(5 - \sqrt{5})/2 \approx 1.382$ competitive ratio, due to a previous lower bound given by \citet{orl/ChenV97}.

The $1.382$ lower bound (for two machines) holds under the assumption that whenever a job is scheduled, it must be processed all the way until its completion.
However, as noted in \citet{tcs/NogaS01}, many applications allow \textbf{restart}: 
a job being processed can be killed (put into pending) and restarted later to make place for a newly-arrived job;
a job is considered completed only if it has been continuously processed on some machine for a period of time that equals to its size.
In other words, whenever a job gets killed, all previous processing of this job is wasted. 

Note that the restart setting is different from the \emph{preemptive} setting, in which the processed portion is not wasted.
\citet{tcs/NogaS01} leave the following open problem:
\emph{Is it possible to beat the $1.382$ barrier with restart?}

In this paper, we bring an affirmative answer to this long-standing open problem.

We propose a variant of the LPT algorithm (with restart) that improves the $1.5$ competitive ratio for the general case, and the $1.382$ competitive ratio for the two-machine case.

\paragraph{Our Replacement Rule.}
A na\"{i}ve attempt for the replacement rule would be to replace a job whenever the newly-arrived job has a larger size.
However, it is easy to observe that the na\"{i}ve attempt fails even on one machine: the worst case competitive ratio is $2$ if we keep replacing jobs that are almost completed (with jobs of slightly larger size). 
Hence we should prevent a job from being replaced if a large portion has been processed.
Moreover, we allow a newly-arrived job to replace a processing job only if it has a much larger size, in order to avoid a long chain of replacements.
As we will show by an example in Section~\ref{sec:candidate_alg}, the worst case competitive ratio is $1.5$ if a job of size $1$ is replaced by a job of size $1+\epsilon$, which is in turn replaced by a job of size $1+2\epsilon$, etc.

We hence propose the following algorithm that applies the above rules.

\paragraph{\lptr.}
As in the LPT algorithm, our algorithm schedules the largest pending job whenever there is an idle machine.
The main difference is that our algorithm may kill a processing job to make place for a newly-arrived job according to the following rule. Upon the arrival of a job $j$, we kill a processing job $k$ (i.e., put $k$ into pending) and schedule $j$ if:
\begin{compactitem}
	\item[1.] $j$ is the largest pending job and $k$ is the smallest among the $m$ processing jobs;
	\item[2.] the processed portion of $k$ is less than $\alpha p_j$;
	\item[3.] the size of $j$ is more than $1 + \beta$ times larger than $k$, i.e., $p_j > (1+\beta) p_k$,
\end{compactitem}
where $0 < \alpha, \beta < \frac{1}{2}$ are parameters of the algorithm.
We call such an operation a \emph{replacement} (i.e. $j$ replaces $k$).

Intuitively, the parameter $\alpha$ provides a bound on the total amount of wasted processing (in terms of the total processing time); while the parameter $\beta$ guarantees an exponential growth in the processing time of jobs if there is a chain of replacements.
With appropriate choice of parameters, we show the following results.

\begin{theorem}\label{theorem:main}
	\lptr, with parameters $\alpha = \frac{1}{200}$ and $\beta = \sqrt{2}-1$, is $(1.5-\frac{1}{20000})$-competitive for the \omm problem with restart.
\end{theorem}

\begin{theorem}\label{theorem:two_machines}
	\lptr, with parameters $\alpha = \beta = 0.2$, is $1.38$-competitive for the \omm problem with restart on two machines.
\end{theorem}

There are many other natural candidate replacement rules.
We list some candidate algorithms that we have considered and their counter examples in Sec~\ref{sec:candidate_alg}.

\paragraph{Our Techniques.}
The main focus of our paper is the general case, i.e., on $m$ machines.
The analysis for the two-machine case is built on the general case by refining some of the arguments.

We adopt an idea from~\citet{orl/ChenV97} to look at the last completed job in our schedule. 
Intuitively, only jobs with size comparable to that of the last job matter.
We develop two kinds of arguments, namely the \emph{bin-packing argument} and the \emph{efficiency argument}.

Assume for contrary that the algorithm has a makespan strictly larger than $1 + \gamma$, where $\gamma := \frac{1}{2}-\frac{1}{20000}$,
we use the bin-packing argument to give an upper bound on the size of the last completed job.
Assume that the last job is large, we will find a number of large jobs that cannot be packed into $m$ bins of size $1$ (recall that we assume $\opt = 1$).
In other words, to schedule this set of jobs, one of the $m$ machines must get a total workload strictly greater than $1$.
For example, finding $2m+1$ jobs of size strictly greater than $\frac{1}{3}$ would suffice.
We refer to such a set of large jobs as an \emph{infeasible} set of jobs.

We then develop an efficiency argument to handle the case when the last job is of small size. 
The central of the argument is a Leftover Lemma that upper bounds the difference of total processing done by the algorithm and by $\opt$.
As our main technical contribution, the lemma is general enough to be applied to all schedules.

Fix any schedule (produced by some algorithm $\alg$) and a time $t$.
Let $\mathcal{M}$ denote the set of machines.
For each machine $M\in\mathcal{M}$, let $\W(M,x)\in\{0,1\}$ be the indicator function of the event that ``at time $x$, machine $M$ is not processing while there are pending jobs''.
Define $\W_t = \sum_{M \in \mathcal{M}}\int_0^t \W(M,x) dx$ to be the total waste (of processing power) before time $t$.
We show (in Section~\ref{sec:leftover}) the following lemma that upper bounds the leftover workload.

\begin{lemma}[Leftover Lemma]\label{lemma:leftover}
	For all time $t$, let $\Delta_t$ be the difference in total processing time before time $t$ between $\alg$ and $\opt$.
	We have $\Delta_t \leq \frac{1}{4}tm + \W_t$.
\end{lemma}

Observe that the total processing power (of $m$ machines) before time $t$ is $tm$.
The Leftover Lemma says that compared to any schedule (produced by algorithm $\alg$), the \emph{extra} processing the optimal schedule can finish before time $t$,
is upper bounded by the processing power wasted by the schedule (e.g., due to replacements),
plus a quarter of the total processing power, which comes from the sub-optimal schedule of jobs.

Consider applying the lemma to the final schedule\footnote{Since a job can be scheduled and replaced multiple times, its start time is finalized only when it is completed.} produced by our algorithm.
Since our algorithm schedules a job whenever a machine becomes idle,
the waste $\W_t$ comes only from the processing (before time $t$) of jobs that are replaced.
Thus (by our replacement) we can upper bound $\W_t$ by $\alpha$ fraction of the total size of jobs that replace other jobs.

We remark that the above bound on the leftover workload is tight for LPT (for which $\W_t = 0$).
Consider $m$ jobs of size $0.5$ arriving at time $0$, followed by $m/2$ jobs of size $1-\epsilon$ arriving at time $\epsilon$.
The optimal schedule uses $\frac{m}{2}$ machines to process the size~$(1-\epsilon)$ jobs and $\frac{m}{2}$ machines to process the size~$0.5$ jobs (two per machine), finishing all jobs at time $1$.
LPT would schedule all the size~$0.5$ jobs first; all of the $\frac{m}{2}$ size~$(1-\epsilon)$ jobs have half of their workload unprocessed at time $1$.
Therefore, the amount of leftover workload at time $t = 1$ is $\frac{m}{4}$.

\paragraph{Other Work.}
The online scheduling model with restart has been investigated in the problem of scheduling jobs on a single machine to maximize the number of jobs completed before their deadlines.
\citet{hoogeveen2000line} study the general case and propose a $2$-competitive algorithm with restart.
Subsequently, \citet{siamcomp/ChrobakJST07} consider the special case when jobs have equal lengths.
They propose an improved $\frac{3}{2}$-competitive algorithm with restart for this special case, and prove that this is optimal for deterministic algorithms.
However, the restart rule and its analysis in our paper do not bear any obvious connections to those in~\citet{hoogeveen2000line} and~\citet{siamcomp/ChrobakJST07} due to the different objectives.

Other settings of the online makespan minimization problem have been studied in the literature.
A classic setting is when all machines are identical and all jobs have release time $0$, but the algorithm must immediately assign each job to one of the machines at its arrival (immediate dispatching). 
This is the same as online load balancing problem.
\citet{DBLP:journals/siamam/Graham69} proves that the natural greedy algorithm that assigns jobs to the machine with the smallest workload is $(2-\frac{1}{m})$-competitive in this setting, which is optimal for $m \le 3$ (due to folklore examples).
A series of research efforts have then been devoted to improving the competitive ratio when $m$ is large (e.g.,~\cite{doi:10.1137/S0097539797324874,BARTAL1995359,KARGER1996400}).
For $m=4$, the best upper bound is 1.7333~\cite{CHEN1994221}, while the best lower bound stands at $1.7321$~\cite{doi:10.1137/S0097539702403438}.
For $m$ that tends to infinity, the best upper bound is $1.9201$ \cite{JOS:JOS54}, while the best lower bound is $1.880$~\cite{JFRUDIN2001}.

A variant of the above setting is that a buffer is provided for temporarily storing a number of jobs; when the buffer is full, one of the jobs must be removed from the buffer and allocated to a machine (e.g.,~\cite{DBLP:conf/imsccs/LiZSC07,DBLP:journals/jco/DosaE10}). \citet{Kellerer1997235} and \citet{DBLP:journals/ipl/Zhang97} use algorithms with a buffer of size one to achieve an improved $4/3$ competitive ratio for two machines. 
\citet{DBLP:journals/siamcomp/EnglertOW14} characterize the best ratio achievable with a buffer of size $\Theta(m)$, where the ratio is between $4/3$ and $1.4659$ depending on the number of machines $m$.
When both preemption and migration are allowed, \citet{DBLP:journals/orl/ChenVW95} give a $1.58$-competitive algorithm without buffer, matching the previous lower bound by \citet{DBLP:journals/ipl/ChenVW94}.
\citet{DBLP:journals/siamdm/DosaE11} achieve a ratio of $4/3$ with a buffer of size $\Theta(m)$.

Finally, if the machines are related instead of identical, the best known algorithm is $4.311$-competitive by \citet{DBLP:journals/jal/BermanCK00}, while the best lower bound is $2$ by \citet{DBLP:journals/orl/EpsteinS00}. 
When preemption is allowed, \citet{DBLP:journals/algorithmica/EbenlendrJS09} show that the upper bound can be improved to $e$. 
For the special case of two related machines, the current best competitive ratio is $1.53$ by \citet{DBLP:conf/soda/EpsteinNSSW99} without preemption, and $4/3$ with preemption by \citet{DBLP:journals/algorithmica/EbenlendrJS09} and \citet{DBLP:journals/orl/WenD98}.

\paragraph{Organization.}
We first provide some necessary definitions in Section~\ref{sec:preli}.
Then we prove the most crucial structural property (Lemma~\ref{lemma:leftover}, the Leftover Lemma) in Section~\ref{sec:leftover}, which essentially gives a lower bound on the efficiency of all schedules.
We present the details of the bin-packing argument and efficiency argument in Section~\ref{sec:general_case}, where our main result Theorem~\ref{theorem:main} is proved.
The special case of two machines is considered in Section~\ref{sec:two_machine}, where Theorem~\ref{theorem:two_machines} is proved.
Finally, we prove in Section~\ref{appendix:restart_hardness} that no deterministic algorithm, even with restart, can get a competitive ratio better than $\sqrt{1.5} \approx 1.225$.

\section{Preliminaries}\label{sec:preli}

Consider the online makespan minimization with $m$ identical machines and jobs arriving over time.
Recall that for each job $j$, $r_j$ denotes its release time and $p_j$ denotes its size.
Let $\opt$ and $\alg$ be the makespan of the optimal schedule and our schedule, respectively.
Recall that we assume without loss of generality that $\opt = 1$ (for analysis purpose only).
Hence we have $r_j + p_j \leq 1$ for all jobs $j$.
Further, let $s_j$ and $c_j := s_j + p_j$ denote the start and completion time of job $j$, respectively, in the \textbf{final} schedule produced by our online algorithm.
Note that a job can be scheduled and replaced multiple times.
We use $s_j(t)$ to denote the last start time of $j$ before time $t$.

We use $n$ to denote the job that completes last, i.e., we have $\alg = c_n = s_n + p_n$.

We consider the time horizon as continuous, and starts from $t=0$.
Without loss of generality (by perturbing the variables slightly), we assume that all $r_i$'s, $p_i$'s and $s_i$'s are different.

\begin{definition}[Processing Jobs]
	For any $t \leq \alg$, we denote by $J(t)$ the set of jobs that are being processed at time $t$, including the jobs that are completed or replaced at $t$ but excluding the jobs that start at $t$.
\end{definition}

Note that $J(t)$ is defined based on the schedule produced by the algorithm at time $t$.
It is possible that jobs in $J(t)$ are replaced at or after time $t$.

\paragraph{Idle and Waste.}
We say that a machine is \emph{idle} in time period $(a,b)$, if for all $t\in(a,b)$, the machine is not processing any job according to our algorithm, and there is \textbf{no} pending job.
We call time $t$ idle if there exists at least one idle machine at time $t$.
Whenever a job $k$ is replaced by a job $j$ (at $r_j$), we say that a \emph{waste} is created at time $r_j$.
The size of the waste is the portion of $k$ that is (partially) processed before it is replaced.
We can also interpret the waste as a time period on the machine.
We say that the waste \emph{comes from} $k$, and call $j$ the \emph{replacer}.

\begin{definition}[Total Idle and Total Waste]
	For any $t\in[0,1]$, 
	define $\I_t$ as the \emph{total idle time} before time $t$, i.e., the summation of total idle time before time $t$ on each machine.
	Similarly, define $\W_t$ as the \emph{total waste} before time $t$ in the final schedule, i.e., the total size of wastes located before time $t$, where if a waste crosses $t$, then we only count its fractional size in $[0,t]$.
\end{definition}

\section{Bounding Leftover: Idle and Waste} \label{sec:leftover}

In this section, we prove Lemma~\ref{lemma:leftover}, the most crucial structural property.
Recall that we define $\W_t$ as the total waste located before time $t$.
For applying the lemma to general scheduling algorithms, (recall from Section~\ref{sec:intro}) $\W_t$ is defined as $\sum_{M \in \mathcal{M}}\int_0^t \W(M,x) dx$, the total time during which machines are not processing while there are pending jobs.
It is easy to check that the proofs hold under both definitions.
We first give a formal definition of the \emph{leftover} $\Delta_t$ at time $t$.

\begin{definition}[Leftover]
	Consider the final schedule and a fixed optimal schedule $\opt$.
	For any $t\in[0,1]$, let $\Delta_t$ be the total processing $\opt$ does before time $t$, minus the total processing our algorithm does before time $t$.
\end{definition}

Since the optimal schedule can process a total processing at most $m(1-t)$ after time $t$, we have the following useful observation.

\begin{observation}\label{observation:processing_after_t}
	The total processing our algorithm does after time $t$ is at most $m(1-t)+\Delta_t$.
\end{observation}

We call time $t$ a \emph{marginal idle time} if $t$ is idle and the time immediately after $t$ is not.
We first define $A_t$, which is designated to be an upper bound on the total processing that could have been done before time $t$, i.e., the leftover workload due to sub-optimal schedule.

\begin{definition}[$\A_t$] \label{definition:A_t}
	For all $t\in[0,1]$, if there is no idle time before $t$, then define $\A_t = 0$, otherwise	
	let $t'\leq t$ be the last idle time before $t$.
	Define $\A_t = \sum_{j \in J(t')} \min\{\delta_j, p_j\}$, where $\delta_j := |T_j| = |\{ \theta\in [r_j,t']:\text{job $j$ is pending at time }\theta \}|$ is the total pending time of job $j\in J(t')$ before time $t'$.
\end{definition}

We show the following claim, which (roughly) says that the extra processing $\opt$ does (compared to $\alg$) before time $t$, is not only upper bounded by total idle and waste ($\I_t + \W_t$), but also by the total size or pending time of jobs currently being processed ($\A_t + \W_t$).

\begin{claim}\label{claim:leftover_leq_min}
	We have $\Delta_t \leq \min\{\A_t, \I_t\} + \W_t$ for all $t\in[0,1]$.
\end{claim}
\begin{proof}
	First observe that we only need to prove the claim for marginal idle times, as we have $\frac{d\Delta_t}{dt}\leq \frac{d\W_t}{dt}$ (while $\frac{d\A_t}{dt} = \frac{d\I_t}{dt} = 0$) for non-idle time $t$.
	Now suppose $t$ is a marginal idle time.
	
	It is easy to see that $\Delta_t$ is at most $\I_t + \W_t$, the total length of time periods before $t$ during which the algorithm is not processing (in the final schedule).	
	Next we show that $\Delta_t \leq \A_t + \W_t$.
	
	Let $\Delta_t(t)$, $\A_t(t)$ and $\W_t(t)$ be the corresponding variables when the algorithm is run until time $t$.
	Observe that for a job $j\in J(t)$, if it is replaced after time $t$, then it contributes a waste to $\W_t$ but not to $\W_t(t)$.
	Moreover, it has the same contribution to $\Delta_t - \Delta_t(t)$ and to $\W_t$.
	Thus we have $\W_t-\W_t(t) = \Delta_t - \Delta_t(t)$. 
	By definition we have $\A_t = \A_t(t)$.
	
	Hence it suffices to show that $\Delta_t(t) \leq \A_t + \W_t(t)$.
	
	Since $t$ is idle, there is no pending job at time $t$.
	Thus the difference in total processing at time $t$, i.e., $\Delta_t$, must come from the difference (between $\alg$ and $\opt$) in processing of jobs in $J(t)$ that has been completed.
	For each $j\in J(t)$, the extra processing $\opt$ can possibly do on $j$ (compared to $\alg$) is at most $\min \{s_j(t)-r_j,p_j\}$. 
	Hence we have $\Delta_t(t) \leq \sum_{j\in J(t)}\min\{s_j(t)-r_j,p_j\}$.
	
	Recall by Definition~\ref{definition:A_t}, we have $T_j \subset [r_j, s_j(t))$ is the periods during which $j$ is pending.
	
	Thus at every time $t\in [r_j,s_j(t))\setminus T_j$, $j$ is being processed (and replaced later).
	Hence $\big|[r_j,s_j(t))\setminus T_j\big|$ is at most the total wastes from $j$ that are created before $s_j(t)<t$,
	which implies 
	\begin{align*}
		\Delta_t(t) \leq \sum_{j\in J(t)}\min\{s_j(t)-r_j,p_j\} \leq \sum_{j\in J(t)}\min\{\delta_j,p_j\} + \W_t(t) = \A_t + \W_t(t),
	\end{align*}
	as desired.
\end{proof}

We prove the following technical claim.

\begin{claim}\label{claim:1/4}
	For any integer $k\geq1$, given any three sequences of positive reals $\{a_i\}_{i\in[k]}$, $\{b_i\}_{i\in[k]}$ and $\{h_i\}_{i\in[k]}$ satisfying conditions
	\begin{compactitem}
		\item[(1)] $0\leq h_1\leq h_2\leq \ldots\leq h_k\leq 1$;
		\item[(2)] for all $j\in[k]$, we have $\sum_{i \in [j]} a_i h_i \geq \frac{1}{4}\sum_{i\in [j]}(a_i+b_i)$,
	\end{compactitem}
	we have $\sum_{i\in[k]} b_i(1-h_i) \leq \frac{1}{4} \sum_{i\in [k]}(a_i+b_i)$.
\end{claim}
\begin{proof}
	We prove the claim by induction on $k$.
	We first show that the claim holds true when $k=1$.
	Note that we have $a_1 h_1 \cdot b_1(1-h_1) \leq (\frac{a_1+b_1}{2})^2 \cdot (\frac{h_1 + (1-h_1)}{2})^2 = \frac{(a_1+b_1)^2}{16}$. Combine with property (2) we know that $b_1(1-h_1) \leq \frac{1}{4}(a_1+b_1)$.
	\begin{figure}[H]
		\centering
		\includegraphics[width = 0.4\textwidth]{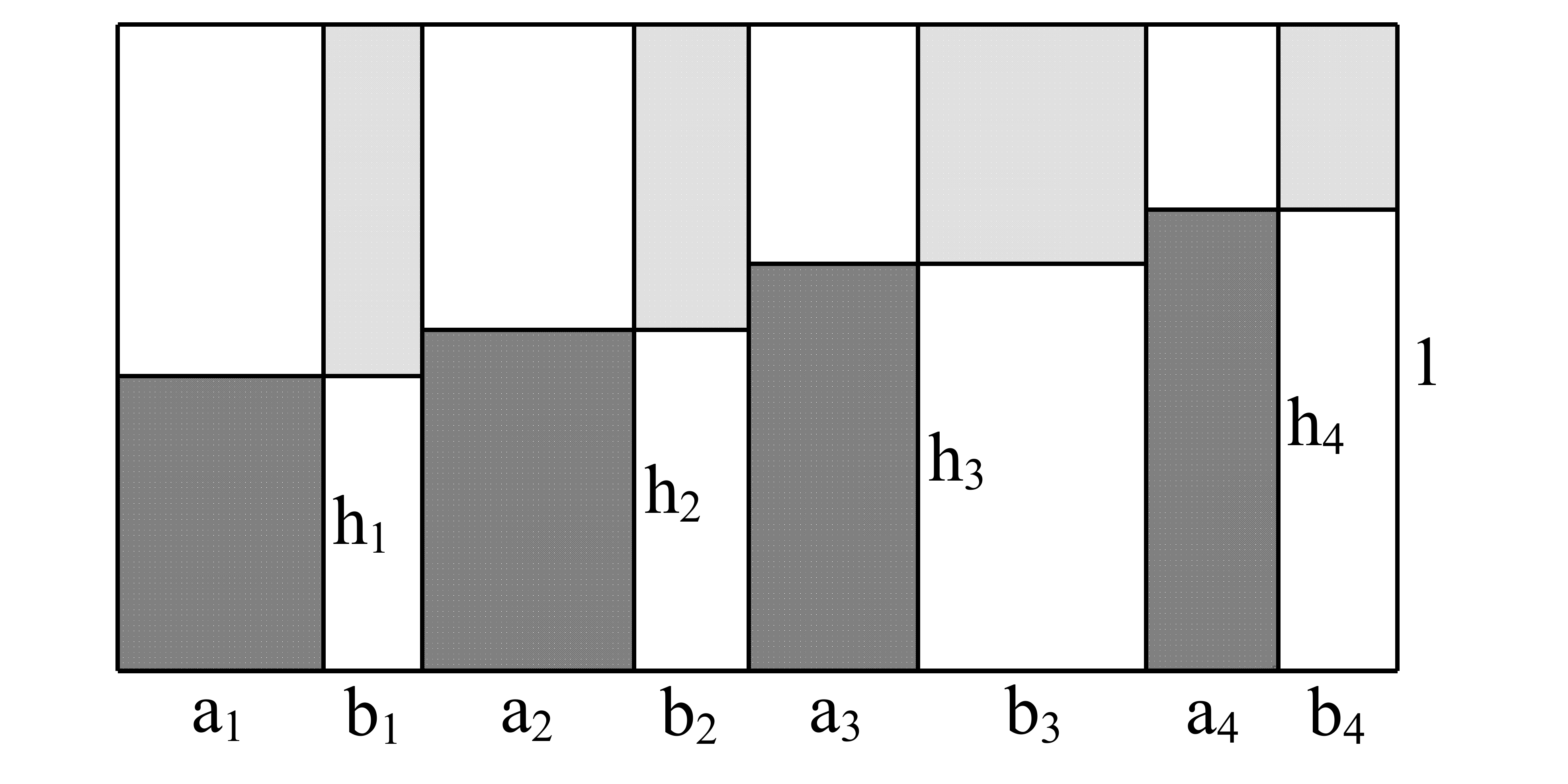}
		\vspace*{-10pt}
		\caption{graph representation of Claim~\ref{claim:1/4} for $k=4$}
	\end{figure}
	
	Now suppose the claim is true for all values smaller than $k$.
	Using induction hypothesis on $\{a_i\}_{i \in  [k-1]}, \{b_i\}_{i \in [k-1]}$ and $\{h_i\}_{i \in [k-1]}$, we have
	\begin{equation*}
		\sum_{i\in[k-1]} b_i(1-h_i) \leq \frac{1}{4} \sum_{i\in [k-1]}(a_i+b_i) \leq \sum_{i\in[k-1]} a_i h_i.
	\end{equation*}
	
	Define $\phi = \min\{b_k, \sum_{i\in[k-1]}(4a_i\cdot h_i-  a_i - b_i)\}$.
	Let $b'_k = b_k - \phi$ and $b'_{k-1} = b_{k-1}+\phi$.
	
	Note that $\{a_i\}_{i \in  [k]}, \{b_i\}_{i \in [k-2]}\cup\{b'_{k-1},b'_k\}$ and $\{h_i\}_{i \in [k]}$ (and their prefixes) satisfy the conditions of the claim:
	first, by definition we have $b'_k >0$ and $b'_{k-1}\geq b_{k-1}>0$; second, since $\{a_i\}_{i \in  [k]}$ and $\{h_i\}_{i \in [k]}$ are not changed, and $b'_{k-1}+b'_k=b_{k-1}+b_k$, if suffices to check condition~(2) for $j = k-1$:
	\begin{equation*}
		\frac{1}{4}\sum_{i\in[k-2]} (a_i+b_i) + \frac{1}{4}(a_{k-1}+b'_{k-1})
		\leq \frac{1}{4}\sum_{i\in[k-1]} (a_i+b_i) + \sum_{i\in[k-1]}(a_i\cdot h_i-  \frac{a_i+b_i}{4})
		= \sum_{i\in[k-1]}(a_i\cdot h_i).
	\end{equation*}
	
	Applying the induction hypothesis on $\{a_i\}_{i \in  [k-1]}, \{b_i\}_{i \in [k-2]}\cup\{b'_{k-1}\}$ and $\{h_i\}_{i \in [k-1]}$,
	\begin{equation*}
		\sum_{i\in[k-1]} b_i(1-h_i) + \phi(1-h_{k-1}) \leq \frac{1}{4} (\sum_{i\in [k-1]}(a_i+b_i) + \phi) \leq \sum_{i\in[k-1]} a_i h_i.
	\end{equation*}
	
	If $\phi = b_k$, then immediately we have 
	\begin{align*}
		\sum_{i\in[k]} b_i(1-h_i) &\leq \sum_{i\in[k-1]} b_i(1-h_i) + b_k(1-h_{k-1}) \leq \frac{1}{4} (\sum_{i\in [k-1]}(a_i+b_i) + b_k) < \frac{1}{4} \sum_{i\in [k]}(a_i+b_i),
	\end{align*}
	as desired.
	Otherwise we have $\phi = \sum_{i\in[k-1]}(4a_i\cdot h_i-  a_i - b_i)$, and hence we have
	\begin{equation*}
		a_k\cdot h_k \geq \frac{1}{4}\sum_{i\in[k]}(a_i+b_i) - \sum_{i\in[k-1]}a_i\cdot h_i = \frac{1}{4}(a_k+b_k+\phi) = \frac{1}{4}(a_k + b'_k),
	\end{equation*}
	which implies $b'_k(1-h_k) \leq \frac{1}{4}(a_k + b'_k)$. Hence we have
	\begin{align*}
		& \sum_{i\in[k]} b_i(1-h_i) = \sum_{i\in[k-2]} b_i(1-h_i) + b'_{k-1}(1-h_{k-1}) + b'_k(1-h_{k}) + \phi(h_{k-1}-h_k)\\
		\leq &\frac{1}{4}\sum_{i\in[k-2]}(a_i+b_i) + \frac{a_{k-1}+b'_{k-1}}{4} + \frac{a_k + b'_k}{4} = \frac{1}{4}\sum_{i\in[k]}(a_i+b_i),
	\end{align*}
	which completes the induction.
\end{proof}

Given Claim~\ref{claim:1/4}, we are now ready to proof the Leftover Lemma.

\begin{proofof}{Lemma~\ref{lemma:leftover}}
	As before, it suffices to prove the lemma for marginal idle times, as we have $\frac{d\Delta_t}{dt}\leq \frac{d\W_t}{dt}$ (while $\frac{d(\frac{1}{4}tm)}{dt} > 0$) for non-idle time $t$.
	Now suppose $t$ is a marginal idle time.
	As before, let $\Delta_t(t)$ and $\W_t(t)$ be the values of variables when the algorithm is run until time $t$.
	
	We prove a stronger statement that $\Delta_t(t) \leq \frac{1}{4}tm +\W_t(t)$, by induction on the number $k$ of marginal idle times at or before time $t$.
	Note that the stronger statement implies the lemma, as we have $\Delta_t - \Delta_t(t) = \W_t- \W_t(t)$.	
	
	In the following, we use a weaker version of Claim~\ref{claim:leftover_leq_min}: we only need $\A_t \leq \sum_{j \in J(t)}\delta_j$.
	
	\textbf{Base Case: $k=1$.}
	Since $t$ is the first marginal idle time, let $g$ be the first idle time, we know that $[g,t]$ is the only idle period.
	Define $J := \{j\in J(t):s_j(t) \leq g\}$ to be the set of jobs that are processed from time $g$ to $t$.
	By definition we have $\I_t \leq (t-g)(m-|J|)$.
	Recall that $\A_t \leq \sum_{j\in J(t)} \delta_j$, where $\delta_j$ is the total pending time of job $j$ before time $t$.
	Hence we have $\delta_j \leq g$ if $j\in J$, and $\delta_j = 0$ otherwise.
	By Claim~\ref{claim:leftover_leq_min} we have
	\begin{align*}
		\Delta_t(t) \leq &\min\{\A_t,\I_t\}+\W_t(t) \leq \min\{ g|J|, (t-g)(m-|J|)\} + \W_t(t) \leq \frac{1}{4}tm + \W_t(t).
	\end{align*}
	
	\textbf{Induction.}
	Now suppose the statement holds for all marginal idle times $0 < t_1 < t_2 < \ldots < t_{k-1}$, and consider the next marginal idle time $t_k$.
	We show that $\Delta_{t_k}(t_k) \leq \frac{1}{4}t_k m +\W_{t_k}(t_k)$.
	First of all, observe that the difference in $\Delta_{t_k}(t_k)$ and $\Delta_{t_j}(t_j)$ must come from the idle periods in $[t_j,t_k]$ and wastes created in $[t_j,t_k]$.
	Hence for all $j<k$ we have
	\begin{equation*}
		\Delta_{t_k}(t_k) - \Delta_{t_j}(t_j) \leq (\I_{t_k} - \I_{t_j}) + \W_{t_k}(t_k) - \W_{t_j}(t_j).
	\end{equation*}
	Hence, if there exists some $j$ such that $\I_{t_k} - \I_{t_j} \leq \frac{1}{4}(t_k-t_j)m$, then by induction hypothesis, 
	\begin{equation*}
		\Delta_{t_k}(t_k) \leq \Delta_{t_j}(t_j) + \frac{1}{4}(t_k-t_j)m + \W_{t_k}(t_k) - \W_{t_j}(t_j)
		\leq \frac{1}{4}t_km + \W_{t_k}(t_k),
	\end{equation*}
	and we are done. Now suppose otherwise.
	
	For all $i\leq k$, let $g_i\in (t_{i-1},t_i)$ be the first idle time after $t_{i-1}$ (assume $g_0 = t_0 = 0$), i.e., $[g_1,t_1], [g_2, t_2], \ldots, [g_k,t_k]$ are the disjoint idle periods.
	Define $J_i := \{j \in J(t_k): r_j \in [t_{i-1}, g_i) \}$.
	Note that for all $j\in J_i$, we have $\delta_j\leq \sum_{x=i}^k (g_x - t_{x-1})$, as $j$ is not pending during idle periods; for $j\in J(r_k)\setminus \cup_{i\leq k}J_i$, we have $\delta_i = 0$.
	For all $i\in[k]$, define
	\begin{equation*}
		\textstyle a_i := t_{k-i+1}-g_{k-i+1}, b_i := g_{k-i+1}-t_{k-i}, h_i := 1-\frac{1}{m}\sum_{x\in[k-i+1]}|J_x|.
	\end{equation*}
	
	We show that the three sequences of positive reals $\{a_i\}_{i\in[k]}$, $\{b_i\}_{i\in[k]}, \{h_i\}_{i\in[k]}$ satisfy the conditions of Claim~\ref{claim:1/4}, which implies $\Delta_{t_k}(t_k) \leq \A_{t_k}+\W_{t_k}(t_k) \leq \frac{1}{4}t_k m + \W_{t_k}(t_k)$ as
	\begin{align*}
		\A_{t_k} & \leq \sum_{j\in J(r_k)} \delta_j \leq \sum_{i\in[k]}\left(\sum_{x=i}^k(g_x-t_{x-1})\right)|J_i|
		= \sum_{i\in[k]} \left((g_i -t_{i-1})\cdot \sum_{x\in[i]}|J_x|\right)\\
		& = m \sum_{i\in[k]} (b_{k-i+1}\cdot (1-h_{k-i+1})) \leq \frac{1}{4}m\sum_{i\in[k]}(a_i+b_i) = \frac{1}{4}t_k m.
	\end{align*}
	Finally, we check the conditions of Claim~\ref{claim:1/4}. Condition~(1) trivially holds.
	For condition~(2), observe that $\I_{t_i} - \I_{t_{i-1}} \leq (t_i - g_i)\cdot (m- \sum_{x\in[i]}|J_x|) = a_{k-i+1}\cdot h_{k-i+1}\cdot m$.
	Hence we have
	\begin{align*}
		\sum_{i\in[j]}(a_i\cdot h_i) &\geq \frac{1}{m} \sum_{i\in[j]}(\I_{t_{k-i+1}} - \I_{k-i}) = \frac{1}{m}(\I_{t_k} - \I_{t_{k-j}})
		> \frac{1}{4}(t_k - t_{k-j}) = \frac{1}{4}\sum_{i\in[j]}(a_i+b_i),
	\end{align*}
	as required.
\end{proofof}

\section{Breaking $1.5$ on Identical Machines} \label{sec:general_case}

In this section, we prove Theorem~\ref{theorem:main}.
We will prove by contradiction: assume for contrary that $\alg > 1 + \gamma$, we seek to derive a contradiction, e.g., no schedule could complete all jobs before time $1$
(Recall that we assume $\opt = 1$).
To do so, we introduce two types of arguments:
we use a bin-packing argument to show that the last job must be of small size, as otherwise there exists a set of infeasible large jobs;
then we use an efficiency argument (built on the Leftover Lemma) to show that the total processing (excluding idle and waste periods) our algorithm completes exceed $m$, the maximum possible processing $\opt$ does.

For convenience of presentation, in the rest of the paper, we adopt the \emph{minimum counter-example} assumption~\cite{tcs/NogaS01}, i.e., we consider the instance with the minimum number of jobs such that $\alg > 1+\gamma$ and $\opt = 1$.
As an immediate consequence of the assumption, we get that no job arrives after $s_n$.
This is because such jobs do not affect the start time of $n$ and therefore could be removed to obtain a smaller counter example.

Recall that in our algorithm, we set $\alpha = \frac{1}{200}$ and $\beta = \sqrt{2}-1$.
Define $\gamma := \frac{1}{2}-\epsilon$, where $\epsilon = \frac{1}{20000}$.
We first provide some additional structural properties of our algorithm, which will be the building blocks of our later analysis.

\subsection{Structural Properties}\label{sec:structural}

Observe that if a job is replaced, then it must be the minimum job among the $m$ jobs that are currently being processed.
Hence immediately we have the following lemma, since otherwise we can find $m+1$ jobs (including the replacer) of size larger than $\frac{1}{2}$.

\begin{fact}[Irreplaceable Jobs]\label{fact:irreplaceable_large_job}
	Any job with size at least $\frac{1}{2}$ cannot be replaced.
\end{fact}

Next, we show that if a job $i$ is pending for a long time, then each of the jobs processed at time $s_i$ must be of (relatively) large size.

\begin{lemma} \label{lemma:boundary_jobs}
	For any job $i$, we have $p_j > \min\{ s_i - r_i, p_i \}$ for all $j\in J(s_i)$.
\end{lemma}
\begin{proof}
	It suffices to consider the non-trivial case when $s_i > r_i$.
	Consider any $j\in J(s_i)$.
	If $p_j > p_i$ or $s_j(s_i) < r_i$\footnote{Note that we use $s_j(s_i)$ here instead of $s_j$ as $j$ can possibly be replaced after time $s_i$.}, then we have $p_j > \min\{ s_i - r_i, p_i \}$.
	Otherwise, we consider time $s_j(s_i)$, at which job $j$ is scheduled.
	Since $p_j < p_i$ and $s_j(s_i) > r_i$ ($i$ is already released), we know that $p_i$ must be processed at $s_j(s_i)$. 
	Hence we know that $i$ is replaced during $(s_j(s_i), s_i)$, which is impossible since $j$ (which is of smaller size than $i$) is being processed during this period.
\end{proof}

Specifically, since $\alg = s_n+p_n > 1+\gamma$ and $r_n+p_n \leq \opt = 1$, we have $s_n - r_n > \gamma$.
Applying Lemma~\ref{lemma:boundary_jobs} to job $n$ gives the following.

\begin{corollary}[Jobs Processed at Time $s_n$] \label{corollary:boundary_jobs_n}
	We have $p_j > \min\{ \gamma, p_n \}$ for all $j\in J(s_n)$.
\end{corollary}

In the following, we show two lemmas, one showing that if a job released very early is not scheduled, then all jobs processed at that time are (relatively) large;
the other showing that if a job is replaced, then the next time it is scheduled must be the completion time of a larger job.

\begin{lemma}[Irreplaceable Jobs at Arrival] \label{lemma:all_jobs_large}
	If a job $k$ is not scheduled at $r_k$ and $r_k < \alpha  p_k$, then $p_j\geq \frac{p_k}{1+\beta}$ for all $j\in J(r_k)$.
\end{lemma}
\begin{proof}
	By our replacement rule, $k$ is not scheduled at $r_k$ either because $k$ is not the largest pending job at $r_k$, or $k$ is the largest pending job, but the minimum job in $J(r_k)$ is not replaceable.
	
	For the second case, since the minimum job $i$ in $J(r_k)$ is processed at most $r_k < \alpha p_k$, job $i$ must violate our third replacement rule, that is $p_i \geq \frac{p_k}{1+\beta}$.
	For the first case, let $k'$ be the first job of size at least $p_k$ that is not scheduled at its release time.
	Then we have $r_{k'} < r_k < \alpha  p_k < \alpha  p_{k'}$. Hence by the above argument every job in $J(r_{k'})$ is of size at least $\frac{p_k}{1+\beta} > r_k$. Thus every job in $J(r_{k})$ is also of size at least  $\frac{p_k}{1+\beta}$.
\end{proof}

\begin{lemma}[Reschedule Rule]\label{lemma:reschedule_rule}
	Suppose some job $k$ is replaced, then the next time $k$ is scheduled must be the completion time of a job $j$ such that $p_k < p_j \leq s_k$.
\end{lemma}
\begin{proof}
	Suppose $k$ is replaced at time $t$ and rescheduled at time $t'$.
	Since $k$ can only replace other jobs at $r_k$, the next time $k$ is scheduled must be when some machine becomes idle. And this happens only if the job $j$ processed before $t'$ on this machine is completed.
	Moreover, since $k$ is pending from $t$ to $t'$, if $s_j \geq t$, i.e., $k$ is pending when $j$ is scheduled, then (by greedy scheduling rule) we have $p_j > p_k$; otherwise $j$ is being processed at time $t$, and we also have $p_j > p_k$ as $k$ is the smallest job among all jobs in $J(t)$ by the replacement rule. Hence, we have $s_k \geq  c_j \geq p_j > p_k$.
\end{proof}

We present the central lemma for our bin-packing argument as follows.
Intuitively, our bin-packing argument applies if there exists time $t_1$ and $t_2$ that are far apart, and the jobs in $J(t_1)$ and $J(t_2)$ are large (e.g. larger than $\frac{1}{3}$):
if $J(t_1)\cap J(t_2) = \emptyset$, then together with job $n$, we have found an infeasible set of $2m+1$ large jobs;
otherwise (since $t_1$ and $t_2$ are far apart) we show that the jobs in $J(t_1)\cap J(t_2)$ must be even larger, e.g. larger than $\frac{2}{3}$.

\begin{lemma}[Bin-Packing Constraints]\label{lemma:impossible_cases}
	Given non-idle times $t_1$, $t_2$ such that $t_1 < t_2$ and $n\notin J(t_1)\cup J(t_2)$, let $a=\min_{j\in J(t_1)}\{p_j\}$ and $b=\min_{j\in J(t_2)}\{p_j\}$, none of the following cases can happen:
	\begin{enumerate}
		\item[(1)] $a>\frac{1}{3}$, $b>\frac{2}{3(1+\beta)}$, $t_2 - t_1 > \frac{2}{3}$ and $p_n > \frac{1}{3}$;
		\item[(2)] $\min\{a,b,p_n\}>\frac{1}{2+\beta}$, and $t_2 - t_1 > 1-\min\{a,b,p_n\}$.
	\end{enumerate}
\end{lemma}
\begin{proof}
	We show that if any of the cases happens, then we have the contradiction that $\opt > 1$.
	
	We first consider case (1).
	We show that we can associate jobs to machines such that every machine is associated with either a job of size larger than $\frac{2}{3}$, or two jobs of size larger than $\frac{1}{3}$.
	Moreover, we show that every job is associated once, while $n$ is not associated.
	Note that since $p_n > \frac{1}{3}$, such an association would imply the contradiction that $\opt > 1$.
	
	First, we associate every $j\in J(t_2)$	to the machine that it is processed on.
	If $p_j > \frac{2}{3}$ then we are done with this machine; otherwise (when $p_j \leq \frac{2}{3}$), we have $s_j(t_2) > t_1$ and we show that we can associate another job of size larger than $\frac{1}{3}$ to this machine.
	\begin{compactitem}
		\item If the job $i\in J(t_1)$ processed on this machine is not replaced, or $p_i\leq\frac{2}{3(1+\beta)}$, then we associate $i$ with this machine (it is easy to check that $i$ has not been associated before);
		\item otherwise the first job that completes after $t_1$ must be of size larger than $(1+\beta) p_i > \frac{2}{3}$, which also has not been associated before. Thus we associate it to this machine.
	\end{compactitem}
	In both cases we are able to do the association, as claimed.
	
	Next we consider case (2).
	Consider any job $i \in J(t_1)$, we have $p_i > \frac{1}{2+\beta}$.
	We apply an association argument similar as before: let $M$ be the machine that job $i$ is processed on.
	\begin{compactitem}
		\item If $i$ is replaced, then we associate the first job that completes on this machine after $i$ is replaced, which is of size larger than $\frac{1+\beta}{2+\beta} > 1-\min\{a,b,p_n\}$, to machine $M$; 
		\item otherwise if $p_i > 1-\min\{a,b,p_n\}$ then we associate $i$ to $M$;
		\item otherwise we know that $i$ completes before $t_2$, and we can further associate to $M$ the job in $J(t_2)$ processed on $M$.
	\end{compactitem}
	It is easy to check that every job is associated at most once.
	Hence every machine is associated with either a job of size larger than $1-\min\{a,b,p_n\}$, or two jobs of size larger than $\min\{a,b,p_n\}$, which (together with $p_n > \frac{1}{2+\beta}$) gives $\opt >1$, a contradiction.
\end{proof}

\subsection{Upper Bounding $p_n$: Bin-Packing Argument} \label{sec:large}

We show in this section how to apply the structural properties from Section~\ref{sec:structural} to provide an upper bound $\frac{1}{2+\alpha}$ on $p_n$.
Recall that we assume $\alg > 1 + \gamma$. We show that if $p_n > \frac{1}{2+\alpha}$, then Lemma~\ref{lemma:impossible_cases} leads us to a contradiction.
We first prove the following lemma (which will be further used in Section~\ref{sec:general_boundary}), under a weaker assumption, i.e., $p_n > \frac{1}{2+\beta}$.

\begin{lemma}\label{lemma:irreplaceable_n}
	If $p_n > \frac{1}{2+\beta}$, then job $n$ is never replaced.
\end{lemma}
\begin{proof}
	Assume the contrary and consider the last time when $n$ is replaced.
	Note that by Fact~\ref{fact:irreplaceable_large_job}, we have $p_n < \frac{1}{2}$.
	Suppose $n$ is replaced by job $l$ at time $r_l$.
	Then by our replacement rule, we have $p_l\geq (1+\beta) p_n$.
	As $r_l + p_l \leq 1$ and $s_n + p_n > 1+\gamma$, we have
	\begin{equation*}
		s_n - r_l > (1+\gamma-p_n) - (1-p_l) > \gamma+\beta p_n > 1-p_n > p_n,
	\end{equation*}
	where the second last inequality holds since $\gamma > \frac{1}{2+\beta}$ and $p_n > \frac{1}{2+\beta}$.
	Since $r_n < r_l$, by Lemma~\ref{lemma:boundary_jobs}, we have $\min_{j\in J(s_n)}\{p_j\}>\min\{s_n-r_n, p_n\} =p_n$.
	Thus we can apply Lemma~\ref{lemma:impossible_cases}(2), with $t_1 = r_l$, $t_2 = s_n$, 
	$a=\min_{j\in J(t_1)}\{p_j\}=p_n$ and $b=\min_{j\in J(t_2)}\{p_j\}>p_n$, and derive a contradiction.
\end{proof}

\begin{lemma}[Upper Bound on Last Job]\label{lemma:large}
	We have $p_n \leq \frac{1}{2+\alpha}$.
\end{lemma}
\begin{proof}
	We first show a weaker upper bound: $p_n \leq \frac{1+\beta}{2}$.
	Assume the contrary that $p_n > \frac{1+\beta}{2} > \frac{1}{2}$.
	As shown in the proof of Lemma~\ref{lemma:boundary_jobs}, for all $j\in J(s_n)$, if $s_j > r_n$, we have $p_j > p_n > \frac{1+\beta}{2} > 1-\gamma$; otherwise $s_j < r_n$ and we have $p_j > s_n - r_n \geq (1+\gamma-p_n) - (1-p_n) = \gamma$.
	Among the $m+1$ jobs $J(s_n)\cup\{n\}$, there exist two jobs, say $k$ and $j$, that are scheduled on the same machine in $\opt$.
	Moreover, we have $s_k, s_j < r_n$, since otherwise one of them is larger than $1-\gamma$ and they cannot be completed in the same machine within makespan $1$.
	Let $k$ be the one with a smaller release time, i.e., $r_k < r_j$. Then we have $r_k \leq 1- p_k - p_j < 1 - 2\gamma <\alpha p_k$.
	
	Observe that $k$ is never replaced, as otherwise (by Lemma~\ref{lemma:reschedule_rule}, the reschedule rule) we have $s_k > p_k$, which implies 
	$r_n + p_n > s_k + p_n > p_k + p_n > \gamma + \frac{1+\beta}{2} > 1$, a contradiction.
	
	By Lemma~\ref{lemma:all_jobs_large}, we know that $s_k = r_k$ ($k$ is scheduled at its arrival time $r_k$), as otherwise the minimum job in $J(r_k)$ is of size at least $\frac{\gamma}{1+\beta}$.
	Thus we have $s_k \geq \frac{\gamma}{1+\beta}$, which is also a contradiction as
	$r_n + p_n > s_k + p_n > \frac{\gamma}{1+\beta} + \frac{1+\beta}{2} > 1$ (recall that $\gamma = \frac{1}{2}-\epsilon$ and $\beta = \sqrt{2}-1$).
	
	By $r_k + p_k + p_j \leq 1$ and $r_k + p_k + p_n > 1+\gamma$, we have $p_n > 2\gamma$.
	Hence $r_n < 1-2\gamma < \alpha p_n$.
	By Lemma~\ref{lemma:all_jobs_large}, we know that the minimum job in $J(r_n)$ is of size at least $\frac{p_n}{1+\beta} > \frac{1}{2}$.
	Hence we have the contradiction that there are $m+1$ jobs, namely $J(r_n)\cup\{n\}$, of size larger than $\frac{1}{2}$.
	
	Hence we have that $p_n \leq \frac{1+\beta}{2}$. Now assume that $p_n > \frac{1}{2+\alpha}$.
	
	By Lemma~\ref{lemma:irreplaceable_n}, we know that $n$ is never replaced.
	By Corollary~\ref{corollary:boundary_jobs_n}, all jobs in $J(s_n) \cup \{n\}$ are of size at least $\min\{\gamma, p_n\}\geq \frac{1}{2+\alpha}$.
	Let $k,j\in J(s_n) \cup \{n\}$ be scheduled on the same machine in $\opt$ such that
	$r_k \leq 1-p_k-p_j < 1- \frac{2}{2+\alpha} \leq \alpha p_k$.
	Note that different from the previous analysis (when $p_n > \frac{1+\beta}{2}$), it is possible that $n\in \{k,j\}$.
	
	By Lemma~\ref{lemma:all_jobs_large}, if $k$ is not scheduled at $r_k$, then each job $i \in J(r_k)$ has size $p_i \geq \frac{p_k}{1+\beta} \geq \frac{1}{(2+\alpha)\cdot(1+\beta)}>\frac{1}{3}$.
	Then we can apply Lemma~\ref{lemma:impossible_cases}(1) with $t_1 = r_k$ and $t_2 = s_n$ to derive a contradiction (observe that $t_2-t_1 = s_n - r_k > (1+\gamma-p_n) - (1- p_k-p_j) > \frac{3}{2+\alpha} - \frac{1+\beta}{2} \geq \frac{2}{3}$).
	
	Hence we know that $k$ is scheduled at time $r_k$ (thus we conclude that $k \neq n$).
	
	We show that $k$ must be replaced (say, by job $l$ at time $r_l$) , as otherwise by $r_k+p_k+p_n>1+\gamma$ and $r_k+p_k+p_j\leq 1$, we have $p_n > \gamma+p_j > \gamma + \frac{1}{2+\alpha} > \frac{1+\beta}{2}$, which is a contradiction.
	
	By Lemma~\ref{lemma:reschedule_rule}, we have $s_k > p_k \geq \frac{1}{2+\alpha}$.
	We also have that $p_n \leq 1-\frac{1}{2+\alpha}$, as otherwise $r_n \leq 1 - p_n < \frac{1}{2+\alpha} < s_k$ ($k$ is scheduled at $s_k$, when $n$ is pending), which implies $p_k > p_n > 1-\frac{1}{2+\alpha} > \frac{1}{2}$, contradicting Fact~\ref{fact:irreplaceable_large_job} (jobs larger than $\frac{1}{2}$ cannot be replaced).
	Since $p_l > (1+\beta)p_k$, we have
	\begin{align*}
		s_n - r_l &> (1+\gamma-p_n) - (1-(1+\beta) p_k) = (1+\beta) p_k +\gamma - p_n \\
		&> \frac{1+\beta}{2+\alpha} +\gamma - \frac{1+\alpha}{2+\alpha} > \frac{1+\alpha}{2+\alpha}.
	\end{align*}
	Then we can apply Lemma~\ref{lemma:impossible_cases}(2), with $t_1=r_l$, $t_2=s_n$, $a=p_k$, $b \geq \frac{1}{2+\alpha}>\frac{1}{2+\beta}$, to derive a contradiction.
\end{proof}

Given the upper bound $\frac{1}{2+\alpha}$ on $p_n$, we show the following stronger version of Lemma~\ref{lemma:irreplaceable_n} that any job of size larger than $\frac{1}{2+\beta}$ cannot be replaced.

\begin{corollary}[Irreplaceable Threshold]\label{corollary:not_replaceable}
	Any job of size larger than $\frac{1}{2+\beta}$ cannot be replaced.
\end{corollary}
\begin{proof}
	Assume the contrary that some job $j$ of size larger than $\frac{1}{2+\beta}$ is replaced (say, by job $k$ at $r_k$).
	Then we have $p_k>(1+\beta)p_j$ and $\min_{i\in J(r_k)}\{p_i\} = p_j > \frac{1}{2+\beta}$.
	On the other hand, by Corollary~\ref{corollary:boundary_jobs_n}, we have $\min_{i\in J(s_n)}\{p_i\} > \min\{ p_n,\gamma \} > \frac{1}{2+\beta}$.
	Since $p_n < \frac{1}{2+\alpha} < \gamma$, we have	
	\begin{equation*}
		s_n - r_k > (1+\gamma-p_n) - (1-\frac{1+\beta}{2+\beta}) = \gamma - p_n + \frac{1+\beta}{2+\beta} > \frac{1+\beta}{2+\beta}.
	\end{equation*}
	
	Hence we can apply the bin-packing argument to derive a contradiction: by Lemma~\ref{lemma:impossible_cases}(2), with $t_1 = r_k$, $t_2=s_n$, $a=p_j$ and $b>\frac{1}{2+\beta}$, we have a contradiction.
\end{proof}

\subsection{Lower Bounding $p_n$: Efficiency Argument}

Next we establish a lower bound on $p_n$, applying the Leftover Lemma.
First, observe that if $n$ is never replaced, then $n$ is pending from $r_n$ to $s_n$, where $r_n \leq 1-p_n$; if $n$ is replaced, then $n$ is pending from $r_l$ to $s_n$, where $r_l$ is the last time $n$ is replaced.
Since $r_l \leq 1-(1+\beta)\cdot p_n < 1-p_n$, in both cases $n$ is pending during time period $[1-p_n,s_n)$.
Hence we have the following fact.

\begin{fact}[Non-idle Period]\label{fact:no_idle_after_1-p_n}
	Every $t\in [1-p_n,s_n)$ is non-idle.
\end{fact}

As a warm-up, we show the following simple lower bound on $p_n$ using the Leftover Lemma.

\begin{lemma}[Simple Lower Bound]\label{lemma:tiny}
	We have $p_n > \frac{1}{3} - 2\alpha$.
\end{lemma}
\begin{proof}
	Let $t = 1-p_n$.
	Since there is no waste after time $1$, the total waste located after time $t$ is $\W_1 - \W_t$.
	Since no machine is idle during time $[ 1-p_n, s_n )$, the total processing our algorithm does after time $t$ is at least $m(s_n - t) - (\W_1 - \W_t)$.
	On the other hand, the total processing our algorithm does after time $t$ is upper bounded by $m(1-t) + \Delta_t$ (Observation~\ref{observation:processing_after_t}).
	Applying Lemma~\ref{lemma:leftover} (the Leftover Lemma) on $\Delta_t$, we have
	\begin{equation*}
		m(s_n-t) < m(1-t) + (\W_1 - \W_t) + \Delta_t \leq m(1-t) + \W_1 + \frac{1}{4}tm
		\leq m(1-t) + \alpha m + \frac{1}{4}tm.
	\end{equation*}
	
	Note that the last inequality follows since (by our replacement rule) each job $j$ can only create a waste of size at most $\alpha\cdot p_j$,
	and the total size of jobs is at most $m$.
	Thus we have
	\begin{equation*}
		\textstyle	\alg = s_n + p_n \leq 1+\alpha + \frac{1-p_n}{4}+p_n = \frac{5}{4} +\alpha +\frac{3}{4}p_n,
	\end{equation*}
	which implies that  (recall that we assume $\alg > 1+\gamma = \frac{3}{2}-\epsilon$)
	\begin{equation*}
		\textstyle p_n \geq \frac{4}{3}\cdot(\alg - \frac{5}{4} - \alpha) > \frac{4}{3}\cdot(\frac{3}{2}-\epsilon - \frac{5}{4} - \alpha) = \frac{1}{3} - \frac{4}{3}(\epsilon+\alpha) > \frac{1}{3} - 2\alpha,
	\end{equation*}
	where the last inequality holds by our choices of parameter, i.e., $\epsilon = \frac{1}{20000}$ and $\alpha = \frac{1}{200}$.
\end{proof}

Observe that the Leftover Lemma provides tighter upper bounds for smaller values of $t$.
Thus in the above proof, if we can find a smaller $t$ such that there is no (or very little) idle time from $t$ to $s_n$, then we can obtain a stronger lower bound on $p_n$.

\begin{lemma}[Lower Bound on Last Job]\label{lemma:replaceable_n}
	We have $p_n > \frac{1}{2+\beta}$.
\end{lemma}
\begin{proof}
	Assume for contrary that $p_n \leq \frac{1}{2+\beta}$.
	Then by Corollary~\ref{corollary:boundary_jobs_n}, we have $p_j > \min\{\gamma,p_n\} = p_n$ for all job $j\in J(s_n)$.
	Hence at least two jobs $k,j\in J(s_n)\cup\{n\}$ are scheduled on the same machine in $\opt$, such that $r_k < 1-2p_n$.
	We prove the following claim\footnote{We remark that the proof of Claim~\ref{claim:idle_after_r_k} relies on the fact that $p_n$ is not too small. Hence Lemma~\ref{lemma:tiny} (the warm-up lower bound) is necessary for achieving the improved lower bound (Lemma~\ref{lemma:replaceable_n}).}, which enables us to use a refined efficiency argument, i.e., apply the Leftover Lemma on a earlier time $t = r_k$.
	For continuity of presentation, we defer its proof to the end of this subsection.
	\begin{claim}\label{claim:idle_after_r_k}
		The total idle time $\I_{s_n} - \I_{r_k}$ during time period $[r_k, s_n]$ is at most $3 \alpha\cdot m$.
	\end{claim}
	
	By the above claim, the total processing $\alg$ does after time $t = r_k$ is at least $m(s_n -t) - (\W_1 - \W_t) - 3\alpha m$.
	On the other hand, by Observation~\ref{observation:processing_after_t} and Lemma~\ref{lemma:leftover} we have
	\begin{equation*}
		\textstyle m(s_n -t) - (\W_1 - \W_t) - 3\alpha m \leq m(1-t) + \Delta_t \leq m(1-t) + \frac{1}{4}tm + \W_t.
	\end{equation*}
	
	Rearranging the inequality, we have (recall that $t = r_k < 1- 2p_n$)
	\begin{equation*}
		\textstyle 1+\gamma < \alg =s_n+p_n \leq 1+\frac{1-2p_n}{4}+4\alpha +p_n \leq \frac{5}{4} + 4\alpha+\frac{\beta}{2(2+\beta)},
	\end{equation*}
	which is a contradiction by our choice of parameters.
\end{proof}

It remains to prove Claim~\ref{claim:idle_after_r_k}.

\begin{proofof}{Claim~\ref{claim:idle_after_r_k}}
	Recall that we assume $\frac{1}{3}-2\alpha < p_n \leq \frac{1}{2+\beta}$, and there exist two jobs $k,j\in J(s_n)\cup\{n\}$ (of size at least $p_n$) scheduled on the same machine in $\opt$, such that $r_k < r_j$.
	
	Observe that $s_k - r_k > (1+\gamma-p_n-p_k)-(1-p_k-p_j) > \gamma$.
	We first upper bound the total idle time in $[r_k, s_k)$ by $3\alpha\cdot m$, and then show that there is no idle time after time $s_k$ (until time $s_n$).
	
	\begin{lemma}\label{lemma:idle_after_r_k_}
		The total idle time in $[r_k, s_k]$ is at most $3 \alpha\cdot m$.
	\end{lemma}
	\begin{proof}
		As $k$ is released at $r_k$ while (eventually) scheduled at $s_k$, we know that if there is an idle period $[a,b]\subseteq [r_k,s_k)$,
		then $k$ must be processed during time period $[a,b]$ and replaced after time $b$.
		Since $k$ cannot be replaced if it is processed $\alpha$, immediately we have $b-a < \alpha$.
		
		Now consider any fixed machine $M$.
		Suppose $a_1 \geq r_k$ is the first time machine $M$ becomes idle.
		We know that $k$ is being processed on some other machine at time $a_1$ (as no job is pending), and replaced at some $r_l < a_1+\alpha$.
		Then the job processed on $M$ at time $r_l$ must be of size larger than $p_k$, which implies that there is no idle period in $[r_l, a_1+p_k]$.
		
		By Fact~\ref{fact:no_idle_after_1-p_n}, there is no idle time after $1-p_n < \frac{2}{3}+2\alpha$.
		Hence the number of idle periods on machine $M$ during time period $[r_k, s_k)$ is at most $\lceil \frac{\frac{2}{3}+2\alpha}{p_k} \rceil \leq 3$, which implies that the total idle time in $[r_k, s_k)$ (on all $m$ machines) is at most $3\alpha\cdot m$.
	\end{proof}
	
	Next we show that every time $t\in (s_k,s_n)$ is non-idle.
	
	Suppose otherwise, let $t\in (s_k,s_n)$ be the last idle time before $s_n$.
	Since there is no pending job at time $t$, we know that $n$ is either being processed at time $t$, or is not released.
	
	For the first case, since $t<s_n$, we know that $n$ is replaced after time $t$.
	Let $r_l$ be the last time $n$ is replaced, we have $r_l \leq 1-p_l < 1-(1+\beta)p_n$.
	Note that $n$ is pending from $r_l$ to $s_n$, which is of length
	\begin{equation*}
		s_n - r_l > (1+\gamma - p_n) - (1-(1+\beta)p_n) > \gamma + \beta\cdot p_n.
	\end{equation*}
	
	Since $s_k < t < r_l$, the total processing our algorithm does after time $r_k$ is
	$m(s_n-r_l) + m(s_k - r_k) - 4\alpha\cdot m > m$, contradicting $\opt = 1$.
	
	Now we consider the second case, i.e., $s_k < t < r_n$.
	By the same reasoning, we know that $n$ is never replaced, i.e., $n$ is pending from $r_n$ to $s_n$, which is of length larger than $\gamma$.
	
	Recall that we have $s_k - r_k > \gamma$, and since $k\in J(s_n)$, we have $p_k > s_n - r_n > \gamma$.
	
	\begin{claim}
		We have $r_n-s_k > 6\alpha + 2\epsilon$ (by the efficiency argument).
	\end{claim}
	\begin{proof}
		Assume the contrary that $s_n-r_k \leq 1 + 6\alpha$, we know that the total idle time after $s_k$ is at most $r_n - s_k \leq 1+6\alpha -2\gamma =6\alpha+2\epsilon$.
		Then we have (by Observation~\ref{observation:processing_after_t}, Lemma~\ref{lemma:leftover}:
		\begin{align*}
			\alg = s_n + p_n
			& \leq 1 + \alpha+ \frac{1}{4}r_k + 3\alpha + 6\alpha + 2\epsilon + p_n \\
			& \leq 1 + 10\alpha+2\epsilon + \frac{1}{4}(1-p_n-\gamma) + p_n \\
			& \leq 1 + ( 10\alpha +\frac{9}{4}\epsilon + \frac{1}{8} + \frac{3}{4(2+\beta)}) \leq 1+\gamma,
		\end{align*}
		contradicting our assumption that $\alg > 1+\gamma$.
	\end{proof}
	
	On the other hand, we have the following contradicting claim.
	
	\begin{claim}
		We have $r_n-s_k < 6\alpha + 2\epsilon$ (mainly by the bin-packing argument).
	\end{claim}
	\begin{proof}
		Observe that $k$ is processed from $s_k$ to $s_n$.
		Hence we have $p_k \geq s_n - s_k$.
		If $p_k < \frac{1}{2}+6\alpha$, then we are done as $r_n - s_k \leq p_k - (s_n - r_n) < \frac{1}{2}+6\alpha - \gamma = 6\alpha + \epsilon$.
		
		Assume for contrary that $r_n-s_k \geq 6\alpha + 2\epsilon$, we have $p_k \geq \frac{1}{2}+6\alpha$. 
		We apply an association argument similar to Lemma~\ref{lemma:impossible_cases} to show that on every machine $M$, we can find
		\begin{compactitem}
			\item either a job of size larger than $\frac{2}{3}+2\alpha$, or
			\item a job $a\notin J(s_n)$ of size $p_a > \frac{1}{3}+4\alpha$, and a  job $b\in J(s_n)$ (of size $p_b > \frac{1}{3}-2\alpha$),
		\end{compactitem}
		that are completed on $M$, and none of them is $n$.
		
		Let $x\in J(s_k)$ and $y\in J(s_n)$ be processed on $M$.
		By Lemma~\ref{lemma:boundary_jobs} we have $p_x > \gamma > \frac{1}{3}+4\alpha$ and $p_y > p_n > \frac{1}{3}-2\alpha$.
		If $x\neq y$, then we are done with this machine, i.e., $a$ is the first job completed after $s_k$ on $M$ and $b=y$.
		Note that we have $a\neq b$ as either $a=x$ or $a$ is a replacer, which cannot be processed at $s_n > 1$.
		
		If $x=y$, then we are also done if $p_x > \frac{2}{3}+2\alpha$.
		
		Now suppose $p_x \leq \frac{2}{3}+2\alpha$. Since $x$ is processed at $s_n > 1$, we know that the job $z$ processed before $s_x$ must be completed. Moreover, since $k$ (of size $>\frac{1}{2}$) is pending during $[r_k,s_k]$, we have
		\begin{equation*}
			p_z \geq (s_n - r_k) - p_x > 2\gamma + 6\alpha + 2\epsilon - \frac{2}{3} - 2\alpha = \frac{1}{3}+4\alpha.
		\end{equation*}
		Hence we have found $a = z$ and $b = x$ completed on $M$ and the association is completed.
		
		Thus in $\opt$ (which completes all jobs before time $1$), three jobs of size in $(\frac{1}{3}-2\alpha, \frac{1}{3}+4\alpha)$ from $J(s_n)$ are scheduled on the same machine, which means that at least one of them, say $x$, is released before $6\alpha < s_k$.
		Observe that we also have $p_n \leq \frac{1}{3}$, as $\min_{i\in J(s_n)}\{p_i\} > p_n$.
		
		Recall that $t\in (s_k,r_n)$ is the last idle time before $r_n$.
		As there is no pending jobs at time $t$, $x$ must be processed at time $t$, but replaced later (as otherwise $p_x > s_n - r_n > \gamma$).
		Hence we know that $t \leq 1-(1+\beta)p_x < 1-(1+\beta)(\frac{1}{3}-2\alpha)$, which implies that total idle time in $[r_k, s_n]$ can be upper bounded by
		\begin{equation*}
			(t-s_k)m < 1-(1+\beta)(\frac{1}{3}-2\alpha)-\gamma.
		\end{equation*}
		
		Hence we have the following contradiction:
		\begin{align*}
			\alg &= s_n + p_n \leq 1 + \alpha + \frac{1}{4}r_k + (1 - (1+\beta)(\frac{1}{3} - 2\alpha) -\gamma) + p_n \\
			& \leq 1 +\alpha + \frac{1}{4}\cdot \frac{1}{2} + \frac{3}{4}\cdot \frac{1}{3} + (1 - (1+\beta)(\frac{1}{3} - 2\alpha) -\gamma)
			\leq 1+\gamma,
		\end{align*} 
		where the second inequality holds since $r_k \leq 1-p_k-p_j\leq \frac{1}{2}-p_n$, and $p_n \leq \frac{1}{3}$.
	\end{proof}
	
	As the two claims are contradicting, there is no idle time during $(s_k,s_n)$.
\end{proofof}

\subsection{A Hybrid Argument} \label{sec:general_boundary}

We have shown that assuming $\alg > 1+\gamma$, the size of the last job $n$ can be bounded as $\frac{1}{2+\beta} < p_n \leq \frac{1}{2+\alpha}$.
In the remaining part of this section, we use a hybrid argument to show that we can either use the bin-packing argument to find a set of infeasible large jobs; or derive a contradiction using the efficiency argument.

\medskip\noindent
\textbf{General Framework.}
Given that $\frac{1}{2+\beta} < p_n \leq \frac{1}{2+\alpha} < \gamma$, we have $s_n = \alg - p_n > 1$.
Thus we have $s_j\neq r_j$ for all $j\in J(s_n)$ (as they are processed at time $s_n > 1$).
Moreover, by Corollary~\ref{corollary:boundary_jobs_n}, we have $\min_{j\in J(s_n)}\{p_j\} > \min\{\gamma,p_n\} > \frac{1}{2+\beta}$.
In other words, there exists a set $J(s_n)\cup\{n\}$ of large jobs, each of which is never replaced (by Corollary~\ref{corollary:not_replaceable}), and none of them is scheduled at its release time.
We know that at least two of them, say $k$ and $j$ (assume $k$ is released earlier), are scheduled on the same machine in $\opt$.
Since $s_k - r_k > (1+\gamma-p_n-p_k) - (1-p_k-p_j) \geq \gamma$,
we know that $k$ is pending from $r_k$ to $s_k$, which is a period of length $\gamma$.
Thus either our algorithm finishes a lot of processing during this period (then we can use the efficiency argument), or
there are many idle and waste periods during this period (then we can use the bin-packing argument to find another $m$ large jobs, e.g., larger than $\frac{1}{3}$) .

We first show that $r_k$ and $s_n$ cannot be too far apart.

\begin{lemma}\label{lemma:s_n-r_k_less_1}
	We have $s_n - r_k \leq 1$.
\end{lemma}
\begin{proof}
	Assume for contrary that $s_n - r_k > 1$. We show that can apply the bin-packing argument to find a set of infeasible large jobs.
	We apply an association argument as in the proof of Lemma~\ref{lemma:impossible_cases} to show that every machine is associated with either a job of size larger than $\frac{1+\beta}{2+\beta}$, or two jobs of size larger than $\frac{1}{2+\beta}$.
	Moreover, every job is associated at most once, while $n$ (recall that $p_n > \frac{1}{2+\beta}$) is not associated, which implies a contradiction.
	
	Fix any machine $M$.
	Consider $i\in J(r_k)$ and $x\in J(s_n)$ processed on $M$.
	\begin{compactitem}
		\item If $p_i > \frac{1}{2+\beta}$, then we associate $i$ and $x$ (both of them cannot be replaced) to $M$.
		Observe that since $s_n - r_k > 1$, we have $i\neq x$.
		\item if $p_i \leq \frac{1}{2+\beta}$, we consider the job $l$ processed on $M$ after $i$.
		Since $r_k+p_i < s_k$, we know that $l$ starts during $(r_k,s_k)$, hence $p_l > p_k > \frac{1}{2+\beta}$ and is completed on $M$.
		If $l\neq x$ then we associate $l$ and $x$ to $M$; otherwise $p_x > (s_n-r_k) - p_i > \frac{1+\beta}{2+\beta}$ and we can associate $x$ to $M$.
	\end{compactitem}
	In both cases we can associate jobs to machines as claimed, which gives a contradiction.	
\end{proof}

Lemma~\ref{lemma:s_n-r_k_less_1} immediately implies that the following stronger lower bound on $p_n$.
Note that the new lower bound $\frac{1}{2}-3\alpha$ is crucial in the sense that we have $\frac{1}{1+\beta}(\frac{1}{2}-3\alpha) > \frac{1}{3}$ (thus more convenient to use the bin-packing argument), while for the previous lower bound we have $\frac{1}{(2+\beta)(1+\beta)} < \frac{1}{3}$.

\begin{corollary}\label{corollary:p_n_leq_0.5-3esp}
	We have $p_n > \frac{1}{2} - 3\alpha$.
\end{corollary}
\begin{proof}
	If $s_k > r_n$, then we know that there is no idle time between $r_k$ and $s_n$, as during this period, either $k$ or $n$ is pending.
	Hence by Observation~\ref{observation:processing_after_t}, we have $m(s_n - r_k)  < m(1-r_k) + \Delta_{r_k} + (\W_1 - \W_{r_k})$,
	which implies $s_n < 1+\frac{1}{4}r_k + \alpha$ by the Leftover Lemma. Therefore, we have (recall that we have $r_k < 1-p_k-p_n < 1-2p_n$)
	\begin{align*}
		\alg &= s_n + p_n < 1+\alpha + \frac{1}{4}(1-2p_n) + p_n \\
		&\leq 1+\alpha+\frac{1}{4}+\frac{1}{2}(\frac{1}{2}-3\alpha) = 1+(\frac{1}{2}-\frac{\alpha}{2}) \leq 1+\gamma,
	\end{align*}
	contradicting our assumption that $\alg > 1+\gamma$.
	
	If $s_k < r_n$, then we have $p_k > s_n - r_n > \gamma$ (recall that $p_k$ is processed at time $s_n$).
	By Lemma~\ref{lemma:s_n-r_k_less_1}, the total idle time between $s_k$ and $r_n$ is at most $(1-2\gamma)m = 2\epsilon m$.
	Then by the Leftover Lemma,
	\begin{equation*}
		\alg = s_n + p_n < 1+\alpha + \frac{1}{4}(1-\gamma-p_n) + p_n + 2\epsilon
		\leq \frac{3}{2}-\frac{1}{4}(5\alpha - 9\epsilon) \leq 1+\gamma,
	\end{equation*}
	which is also a contradiction.
\end{proof}

Unfortunately, $\frac{1}{2} - 3\alpha$ is still less than $\frac{1}{2+\alpha}$.
Hence it remains to consider the subtle case when $\frac{1}{2} - 3\alpha < p_n \leq \frac{1}{2+\alpha}$.
Note that so far we have proved that there exists two periods, namely $[r_k,s_k]$ and $[r_n,s_n]$, both of length at least $\gamma$, and contain no idle time.

From the proof of Corollary~\ref{corollary:p_n_leq_0.5-3esp}, we observe that depending on whether the two intervals overlap, the analysis can be quite different.
Hence we divide the discussion into two parts.
As we will show later, the central of the analysis is to give strong upper bounds on $\W_1$.

\subsubsection{Overlapping Case: when $s_k \geq r_n$}

Note that in this case, from $r_k$ to $s_n$, the largest pending job is always at least $p_n$.
Hence there is no idle time in $[r_k,s_n]$.
Moreover, every job that starts in $[r_k,s_n]$ must be larger than $p_n$ (and hence cannot be replaced).
First we show that $r_k \geq \alpha p_k$.

Suppose otherwise, then by Lemma~\ref{lemma:all_jobs_large}, we have $p_i>\frac{1}{1+\beta}(\frac{1}{2}-3\alpha) > \frac{1}{3}$ for all $i\in J(r_k)$.
Hence we can apply Lemma~\ref{lemma:impossible_cases}(1) with $t_1 = r_k$, $t_2 = s_n$, $a>\frac{1}{3}$ and $b>p_n>\frac{1}{2}-3\alpha>\frac{2}{3(1+\beta)}$, for which
\begin{equation*}
	t_2 - t_1 = s_n - r_k > (1+\gamma-p_n) - (1-p_k-p_j) > \gamma+ p_n > \frac{2}{3},
\end{equation*}
to derive a contradiction.

\begin{lemma}
	There exists a time $r \leq r_k$, at which the minimum job processed has been processed at least $\alpha\cdot p_k$; moreover, from $r$ to $s_n$, the largest pending job is always at least $p_n$.
\end{lemma}
\begin{proof}
	If $k$ is the largest pending job at $r_k$, then let $x$ be the job of minimum size processed at time $r_k$, i.e., $x = \arg\min_{j\in J(r_k)}\{p_j\}$.
	For the same reason as argued above, we have $p_x\leq \frac{1}{3} < \frac{p_k}{1+\beta}$.
	Since $k$ does not replace $x$, we conclude that $x$ must be processed at least $\alpha p_k$.
	Hence the corollary holds with $r = r_k$.
	
	Otherwise we consider the earliest time $r$ before $r_k$ such that from $r$ to $r_k$, the largest pending job is always at least $p_k$. Note that we must have $r=r_{k'}$ for some job $k'$ of size $p_{k'} > p_k$. Moreover, $k'$ is the largest pending job at $r_{k'}$, but not scheduled.
	Let $x = \arg\min_{j\in J(r_{k'})}\{p_j\}$.
	We have $p_x \leq \frac{1}{3}$, as otherwise $r_k > \frac{1}{3} > 1-p_k-p_j$ is a contradiction.
	Hence by a similar argument as above, the corollary holds with $r = r_{k'}$.
\end{proof}

By the above lemma, every job that starts in $[r, s_n]$ must be larger than $p_n > \frac{1}{2+\beta}$, which cannot be replaced.
Thus, from $r$ to $s_n$, there is no idle time, and if there is any waste, then it must come from the jobs in $J(r)$.

Let $\P$ be the total processing our algorithm does after time $s_n$, and let $\W = \W_1$.
We show that the following upper bound on $\W-\P$.

\begin{lemma}[Total Waste]\label{lemma:bounded_waste}
	We have
	\begin{equation*}
		\W-\P \leq 2\alpha\cdot r m + \frac{1}{2+\beta}(r m - \A_r - \I_r).
	\end{equation*}
\end{lemma}

We first show how to use Lemma~\ref{lemma:bounded_waste} to prove the desired competitive ratio, and defer the proof of lemma (which is long and contains many cases) to Section~\ref{appendix:total_waste}.

\begin{corollary}\label{corollary:boundary_case}
	When $s_k\geq r_n$, we have $\alg \leq 1+\gamma$.
\end{corollary}
\begin{proof}
	First by Observation~\ref{observation:processing_after_t}, we have (since there is no idle time after $r$)
	\begin{equation*}
		m(s_n - r) +\P - (\W - \W_r) \leq m(1-r) + \Delta_r,
	\end{equation*}
	which (by Lemma~\ref{lemma:bounded_waste}) implies
	\begin{align*}
		s_n &\leq 1+\frac{1}{m}(\W-\P + \Delta_r - \W_r)
		\leq 1+2\alpha\cdot r+\frac{r}{2+\beta}+\frac{1}{m}(\Delta_r - \W_r -\frac{\A_r+\I_r}{2+\beta}) .
	\end{align*}
	
	Since $\Delta_r - \W_r \leq \frac{rm}{4}$ (Leftover Lemma) and $\Delta_r- \W_r\leq \min\{\A_r,\I_r\}$ (Claim~\ref{claim:leftover_leq_min}), we have
	\begin{equation*}
		s_n \leq 1+2\alpha r+\frac{r}{2+\beta}+\frac{1}{m}(\frac{2\min\{\A_r,\I_r\} - (\A_r+\I_r)}{2+\beta} + \frac{\beta}{2+\beta}\frac{rm}{4})
		\leq 1+(2\alpha+\frac{4+\beta}{4(2+\beta)})r.
	\end{equation*}
	
	Given $\alpha = \frac{1}{200}$ and $\beta = \sqrt{2}-1$, it is easy to check that $2\alpha+\frac{4+\beta}{4(2+\beta)} < \frac{1}{2.14}$.
	Hence we have
	\begin{align*}
		\alg &= s_n + p_n \leq 1+\frac{1}{2.14}(r + 2p_n) + \frac{0.14}{2.14}p_n
		\leq 1+\frac{1}{2.14}+\frac{0.14}{2.14(2+\alpha)} \leq 1+\gamma,
	\end{align*}
	where in the second inequality we use the fact that $r\leq r_k < 1-2p_n$.
\end{proof}

\subsubsection{Disjoint Case: when $s_k < r_n$}

Note that in this case we have two disjoint time periods, namely $[r_k,s_k]$ and $[r_n,s_n]$, both of length at least $\gamma$, and during which there is a pending job of size at least $p_n$.
Moreover, by Lemma~\ref{lemma:s_n-r_k_less_1}, we have $r_n - s_k \leq 1-2\gamma = 2\epsilon$.

\begin{lemma}\label{lemma:jobs_at_r_n_0.33}
	We have $\min_{j\in J(r_n)}\{p_j\} > \frac{1}{3}$.
\end{lemma}
\begin{proof}
	Let $x$ be the minimum job in $J(r_n)$. Suppose $p_x\leq \frac{1}{3}<\frac{p_n}{1+\beta}$.
	Then we have $s_x(r_n) > s_k$, which means that $x$ is processed at most $r_n-s_k<\alpha p_n$.
	Hence we know that $n$ is not the largest pending job at $r_n$, as otherwise $x$ would be replaced.
	
	Then we consider the first job $i$ of size $p_i>p_n$ that is released in $[s_x(r_n),r_n]$ and not scheduled.
	Since $x$ is the largest pending job at $s_x(r_n)$, we know that $i$ must be the largest pending job at $r_i$. 
	We also know that $\min_{j\in J(r_i)}\{p_i\}\leq p_x < \frac{p_i}{1+\beta}$, and the minimum job $y$ has been processed less than $\alpha p_i$, which is impossible as $i$ should have replaced $y$.
\end{proof}

Let $\P$ be the total processing our algorithm does after time $r_n+\gamma$, and let $\W = \W_1$.
We show the following lemma, which is an analogy (weaker) version of Lemma~\ref{lemma:bounded_waste} in previous section.

\begin{lemma}[Total Waste]\label{lemma:bounded_waste_weaker}
	We have
	\begin{equation*}
		\W - (\W_{r_n} - \W_{s_k}) -\P \leq (\frac{1}{2} + \alpha)\cdot r_k\cdot m+8\alpha^2 m.
	\end{equation*}
\end{lemma}

For continuity of presentation, we defer its proof to Section~\ref{appendix:total_waste}.

\begin{lemma}\label{lemma:separate_r_k_case}
	When $s_k < r_n$, we have $\alg \leq 1+\gamma$.
\end{lemma}
\begin{proof}
	Since there is no idle time in $[r_k,s_k]\cup[r_n,s_n]$, by Observation~\ref{observation:processing_after_t}, we have
	\begin{align*}
		m(1-r_k) &\geq 2\gamma m+\P-(\W - (\W_{r_n} - \W_{s_k}))
		\geq (1-2\epsilon)m - (\frac{1}{2} + \alpha)\cdot r_k\cdot m+8\alpha^2m,
	\end{align*} 
	which implies $r_k \leq \frac{2\epsilon+8\alpha^2}{\frac{1}{2}-\alpha} < \alpha p_k$, contradicting the assumption $r_k\geq \alpha p_k$.
\end{proof}

\section{Bounding Total Waste}\label{appendix:total_waste}

In this section, we prove Lemma~\ref{lemma:bounded_waste} and~\ref{lemma:bounded_waste_weaker}, which give upper bounds on the total waste $\W$.

\begin{proofof}{Lemma~\ref{lemma:bounded_waste}}
	We upper bound the total waste by partitioning it into four parts.
	
	\paragraph{Part-1: wastes created after $r$.}
	If any job $i\in J(r)$ is replaced, then $p_i \leq \frac{1}{2+\beta}$ and hence can only be scheduled after $s_n$.
	Hence this part of wastes can be upper bounded by $\sum_{i\in J(r): s_i>s_n}p_i$.
	
	\paragraph{Part-2: wastes created by jobs never processed from $r$ to $s_n$.}
	As $s_n - r > 1 - r$, the total processing time of jobs that are never processed in $[r, s_n]$ is at most $r\cdot m$.
	Hence the total waste created by these jobs is at most $\alpha\cdot r\cdot m$.
	
	Note that excluding \textbf{Part-1} and \textbf{Part-2}, the wastes created by jobs that are ever scheduled from $r$ to $s_n$ can only be created by jobs in $J(r)$, as other jobs are larger than $\frac{1}{2+\beta}$ and will not be replaced (hence cannot be a replacer before $r$).
	
	\paragraph{Part-3: wastes created by $i\in J(r)$ at $r_i < s_i(r)$.}
	As $i$ is processed at $r$, we know that $i$ is replaced during $[r_i, s_i(r)]$, and rescheduled at $s_i(r)$.
	Note that the waste is at most $\alpha\cdot p_i$, since it is replaced by $i$. Consider the job completed at $s_i(r)$, by Lemma~\ref{lemma:reschedule_rule} we know that it is larger than $p_i$. 
	Thus we can construct a one to one mapping from this kind of wastes to jobs completed before $r$. 
	Moreover, each waste is bounded by $\alpha$ fraction of its image.
	Hence the total waste of this part is at most $\alpha\cdot r\cdot m$.
	
	\paragraph{Part-4: wastes created by $i\in J(r)$ at $r_i = s_i(r)$.}
	We will show that this part of wastes can be upper bounded by $\sum_{i\notin J(r): s_i>s_n}p_i + \frac{1}{2+\beta}(r m - \A_r - \I_r)$.
	Let $t$ be the last idle time before $r$. We have $\A_r = \A_t = \sum_{u\in J(t)}\min\{\delta_u, p_u\}$, where $\delta_u$ is the total pending time of $u$ before time $t$.
	
	Consider any waste (of size) $w$ created by $i \in J(r)$ at time $r_i = s_i(r)$ on machine $M$.
	
	We interpret the waste as the time period $[r_i-w,r_i]$ on $M$.
	In the following, we charge the waste to a job of size at least $w$ that is not in $J(r)$ and starts after $s_n$, or to	a set of time periods of total length at least $(2+\beta)w$ located before $r$.
	We show that (over all wastes) every job and time period will be charged at most once, and none of them overlaps with the idle periods.
	Moreover, we show that on every machine $M'$, we can find a non-idle time period before $r$ of length at least $\min\{\delta_u,p_u\}$ that is not charged, where $u\in J(t)$ is the job processed on $M'$.
	(It is easy to check whether the charged time period overlap with the idle periods. However, it is more involved to check the disjointness with pending periods, which is the main focus of our charging argument.) 
	Note that such a charging argument gives an upper bound $\sum_{l \notin J(r):s_l>s_n}p_l + \frac{1}{2+\beta}(r m - \A_r -\I_r)$ on the total wasted in this part.
	
	Let $x$ be the minimum job in $J(r)$, and $y$ be the job where $w$ comes from.
	Recall that we have $p_x \leq \frac{1}{3}$ and $x$ is processed at least $\alpha\cdot p_k$, i.e., $r_x = s_x(r) \leq r - \alpha\cdot p_k$.
	Now we present our charging argument.
	\begin{enumerate}
		\item
		If $r-r_i \geq (1+\beta)w$, then we can charge $w$ to itself, together with the time period $[r_i, r]$ on machine $M$.
		It is easy to see that the charged periods are of total length at least $(2+\beta)w$, and disjoint from the idle periods.
		Define $u(M) \in J(t)$ as the job processed on $M$ at the last marginal idle time before time $r$.
		If $u = i$, then since $s_u(t) = r_u$, we have $\delta_u = 0$; otherwise $s_u(t) \leq r_i-w$ and $\delta_u$ is at most the total length of non-idle periods before $r_i - w$.
		In both cases we find a non-idle time period, namely $T_{u(M')}$, of length at least $\delta_u$ that is not charged.
		
		\item 
		If $p_y < p_x$ (which implies $y\notin J(r)$) and $s_y > s_n$, then let $w'$ be the first waste from $y$ that has not been charged.
		\begin{figure}[h]
			\centering
			\vspace*{-10pt}
			\includegraphics[width = 0.75\textwidth]{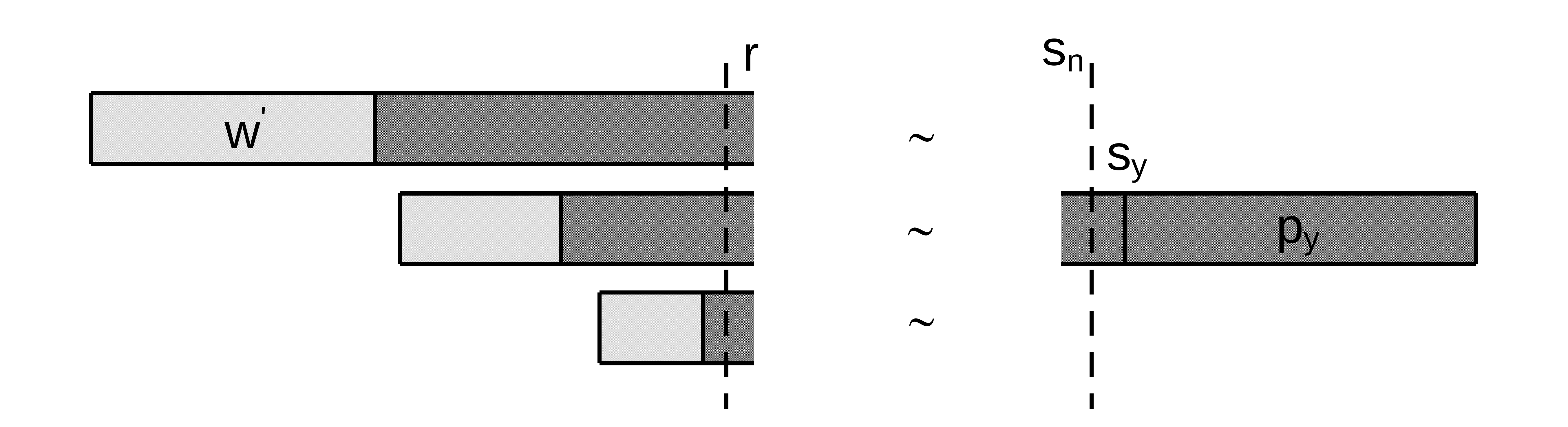}
			\vspace*{-15pt}
		\end{figure}
		
		We charge $w'$ to $p_y > w'$ (processed after $s_n$), and charge the remaining wastes from $y$, which is of total size at most $(1+\beta)w'$, to themselves, $w'$, and the processing (before $r$) of the job that creates $w'$.
		Note that the charged periods are of total length at least $(3+2\beta)w' > (2+\beta)(1+\beta)w'$, and disjoint from the idle periods.
		Since every charged period is either processing of jobs in $J(r)$, or the wastes they create, applying the same argument as above, on every machine $M'$ that has a charged waste, we can find a non-idle time period of length at least $\delta_{u(M')}$ that is not charged.
		
		\item
		If $p_y < p_x$ and $s_y < s_n$, then we know that $y$ is rescheduled (say, on machine $M'$) after $r_i$, and completed before $r$.
		Hence we have $r-r_i > p_y > w$. Moreover, by Lemma~\ref{lemma:reschedule_rule}, the job $z$ completed at $s_y$ on $M'$ must be larger than $p_y$.
		\begin{figure}[H]
			\centering
			\vspace*{-15pt}
			\includegraphics[width = 0.4\textwidth]{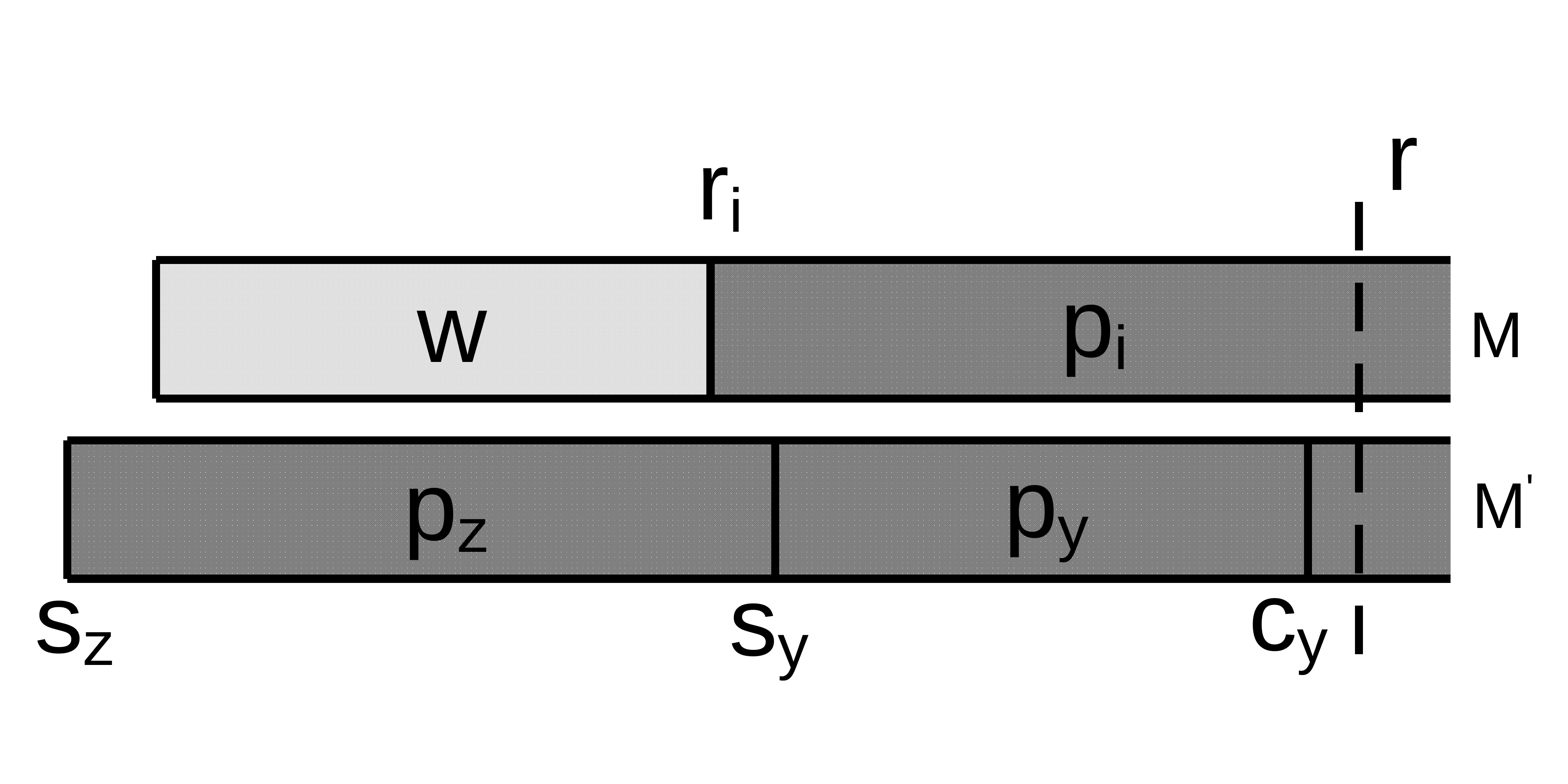}
			\vspace*{-15pt}
		\end{figure}
		
		We charge $w$ to $w$, $[r_i,r]$ on machine $M$, and $[s_y, c_y]$ on machine $M'$.
		The total length of charged periods is at least $w+2p_y>3w$.
		As before, a non-charged non-idle time period of length $\delta_{u(M)}$ can be found on $M$.
		Now consider $u(M')$.
		\begin{compactitem}
			\item If $p_z > p_{u(M')}$, then $[s_z, c_z]$ is the desired non-charged non-idle time on $M'$;
			\item otherwise we know that $u(M')$ is not pending during $[s_y,c_y]$, as $y$ is smaller than $u(M')$ but is scheduled and completed. Hence $T_{u(M')}$ is the desired period on $M'$.
		\end{compactitem}
		
		\item
		If $p_y > p_x$, then we know that $r_i < s_x(r) \leq r-\alpha p_k$, as otherwise $x$ would be replaced instead of $y$.
		Moreover, the next time $y$ is scheduled (say, on machine $M'$) must be before $s_x(r)$, as $x$ cannot be scheduled while $y$ is pending, and we have $p_z > p_y$ for the job $z$ completed before $y$ on $M'$.
		\begin{figure}[H]
			\centering
			\vspace*{-15pt}
			\includegraphics[width = 0.5\textwidth]{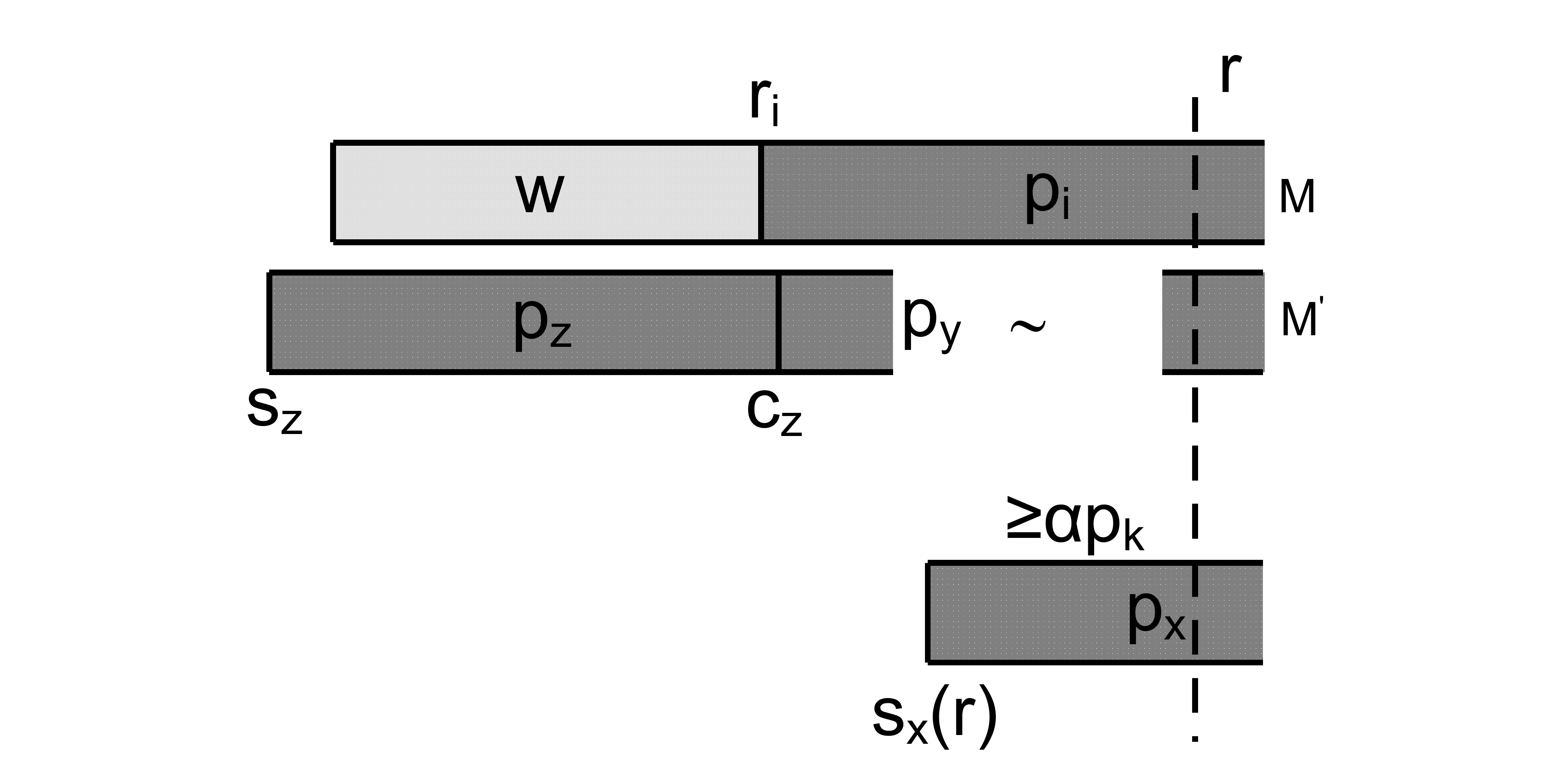}
			\vspace*{-10pt}
		\end{figure}
		
		We charge $w$ to $w$, $[r_i, r]$ on $M$, and $[s_z, r_i]$ on $M'$.
		The charged periods are of total length $w+r-s_z > w + p_z + \alpha p_k > (2+\beta)w$, and are disjoint from idle periods.
		As before, a size $\delta_{u(M)}$ non-charged non-idle time period can be found on $M$.
		Now consider $u(M')$.
		\begin{compactitem}
			\item If $[s_z, r_i]\cap T_{u(M')} = \emptyset$, then $T_{u(M')}$ is the desired period;
			\item otherwise we know that $u(M')\neq z$. Moreover, $u(M')\neq y$, as $y$ is not pending during $[s_z, r_i]$, which means that $u(M')$ is a job processed after $y$. First observe that $y$ cannot be replaced after $c_z$, as otherwise the replacer $l$ must be of size larger than $(1+\beta)p_y > (1+\beta)w > r-r_i$, which implies $\delta_{u(M')} = 0$: either $u_{M'}$ is the replacer, or the replacer of the replacer, etc. Hence $[c_z,c_y]$ is the desired non-charged non-idle period on $M'$: we have $p_y > p_{u(M')}$ (if $u_{M'}$ is ever replaced, then we use Lemma~\ref{lemma:reschedule_rule}, otherwise we use the fact that $u_{M'}$ is pending at $c_z$).
		\end{compactitem}
	\end{enumerate}
	
	Combing the four cases above, we have
	\begin{align*}
		\W \leq &\sum_{i: s_i>s_n}p_i + 2\alpha\cdot r m + \frac{1}{2+\beta}(r m - \A_r -\I_r)
		< \P + 2\alpha\cdot r m + \frac{1}{2+\beta}(r m - \A_r -\I_r),
	\end{align*}
	as claimed.
\end{proofof}

\begin{proofof}{Lemma~\ref{lemma:bounded_waste_weaker}}
	Note that $\W - (\W_{r_n} - \W_{s_k})$ is the total wastes located at $[0,s_k]\cup[r_n,s_n]$.
	Define $\P' = \sum_{j\notin J(r_k) : s_j > s_n}p_j$ to be the total processing of jobs not processed at $r_k$ that start after $s_n$.
	We show that $\W_{r_k}(r_k) - \P' \leq (\frac{1}{2} + 2\alpha)\cdot r_k\cdot m$ and $(\W-\W_{r_n})+(\W_{s_k}-\W_{r_k}(r_k)) \leq \P -\P'$, combing the two upper bounds we have the lemma.
	
	We first show that $(\W-\W_{r_n})+(\W_{s_k}-\W_{r_k}(r_k)) \leq \P -\P'$.
	Note that by definition $(\W-\W_{r_n})+(\W_{s_k}-\W_{r_k}(r_k))$ is at most the total size of wastes created in $[r_k,s_k]\cup[r_n,s_n]$.
	Moreover, any such waste $w$ must come from jobs in $J(r_k)\cup J(r_n)\backslash (J(s_k)\cup J(s_n))$, as otherwise the job is irreplaceable.
	
	Consider any $i\in J(r_n)\backslash (J(s_k)\cup J(s_n))$ processed on machine $M$.
	Let $x\in J(s_n)$ be processed on $M$, then we know that the job completed at $s_x$ must be larger than $\frac{1}{3}$.
	Hence we have $c_x - (r_n+\gamma) > (s_k+\frac{1}{3}+p_x)-(r_n+\gamma) > \frac{1}{3}-\epsilon-3\alpha>2\alpha$, i.e., the contribution of $i$ to $\P-\P'$ is larger than to $(\W-\W_{r_n})+(\W_{s_k}-\W_{r_k}(r_k))$ (note that $i$ can be replaced at most twice).
	
	Moreover, we have $|J(r_n)\backslash (J(s_k)\cup J(s_n))|\leq 4\alpha m$, as otherwise the total processing after $r_k$ is larger than $2\gamma m + 4\alpha(\frac{1}{3}-\epsilon-3\alpha)m - \alpha m > m$, contradicting $\opt =1$.
	This fact will be used later.
	
	Now consider any $i\in J(r_k)\backslash (J(r_n)\cup J(s_k)\cup J(s_n))$.
	Then we have $s_i>s_n$. Hence the contribution of $i$ to $\P-\P'$ is larger than to $(\W-\W_{r_n})+(\W_{s_k}-\W_{r_k}(r_k))$.
	
	Next we show that $\W_{r_k} - \P' \leq (\frac{1}{2} + \alpha)\cdot r_k\cdot m + 9\alpha^2 m$.
	The strategy is similar to the proof of Lemma~\ref{lemma:bounded_waste} (but simpler).
	We partition $\W_{r_k}$ into three parts. 
	
	\paragraph{Part-1: wastes created by $J(r_k)$.}
	Let $R$ be the set of jobs $i\notin J(r_k)$ that replaced some other job at $r_i$.
	We show that $\sum_{i\in R}p_j\leq 4\alpha m$, which implies that the total wastes created by jobs not in $J(r_k)$ is at most $4\alpha^2 m$.
	It suffices to show that on every machine $M$, we can find a set of jobs completed on $M$ that are of total size at least $1-4\alpha$ and not in $R$.
	Fix any machine $M$ and consider $j_1\in J(s_k)$ and $j_2\in J(s_n)$ processed on $M$.
	Note that we have $p_{j_1}>\gamma$, $p_{j_2}>p_n$, and both jobs cannot be replaced.
	If $j_1\neq j_2$, then we find two jobs completed on $M$ of total size larger than $\gamma + p_n > 1-4\alpha$ that are not in $R$;
	otherwise we know that $j_1$ did not replace any job (as it is processed at $s_n > 1$).
	Then we know that the total size of $j_1$ and the job completed at $s_{j_1}$ is larger than $(s_n-s_k)+\min\{ s_k-r_k,\gamma \} > 1-4\alpha$, which completes the analysis.
	
	\paragraph{Part-2: wastes created by $i\in J(r_k)$ at $r_i < s_i(r_k)$.}
	Applying a similar argument as in the proof of Lemma~\ref{lemma:bounded_waste} (Part-3), the total waste of this part can be upper bounded by $\alpha\cdot r_k\cdot m$.
	
	\paragraph{Part-3: wastes created by $i\in J(r_k)$ at $r_i = s_i(r_k)$.}
	Note that there is at most one waste on every machine.
	Consider any waste of size $w$ created at $r_i$ on machine $M$.
	Suppose $w$ comes from job $x$.
	\begin{compactitem}
		\item If $w\leq \frac{r_k}{2}$, then we can charge $w$ to machine $M$;
		\item if $x\in J(r_k)$, then we can charge $w$ to $M$ and the machine that processes $x$ at $r_k$;
		\item if $s_x > s_n$, then the contribution of $x$ to $\W_{r_k} - \P'$ is non-positive;
		\item otherwise we have $x\in J(r_n)\backslash (J(s_k)\cup J(s_n))$. As $|J(r_n)\backslash (J(s_k)\cup J(s_n))|\leq 4\alpha m$, the total waste of this part is at most $4\alpha^2 m$.
	\end{compactitem}
	Hence in total we have
	\begin{equation*}
		\W_{r_k} - \P' \leq 4\alpha^2 m + \alpha\cdot r_k\cdot m + \frac{1}{2}\cdot r_k\cdot m+ 4\alpha^2 m,
	\end{equation*}
	as claimed.
\end{proofof}

\section{Breaking $\frac{5-\sqrt{5}}{2}$ on Two machines with Restart} \label{sec:two_machine}

We prove Theorem~\ref{theorem:two_machines} in this section, that is, we show that running our algorithm with $\beta = \alpha=0.2$ on two machines, i.e., $m=2$, achieves a competitive ratio at most $1.38$, strictly better than the best possible competitive ratio $\frac{1}{2}(5-\sqrt{5}) \approx 1.382$ for the problem without restart on two machines (for deterministic algorithms)~\cite{tcs/NogaS01}.

As before, we adopt the \emph{minimum counter-example} assumption~\cite{tcs/NogaS01}, i.e., we consider the instance with the minimum number of jobs such that $\alg > 1.38$ and $\opt = 1$, and proceed to derive a contradiction.

\subsection{Overview of Techniques} \label{sec:two-machine-overview}

Our analysis for the two-machine case follows the same framework as the general case: we use a bin-packing argument to upper bound the size of last job, and use an efficiency argument to lower bound the size of last job, and finally use a hybrid argument to handle the boundary case.
In this section, we will overview two technical ingredients that are specifically developed for the two-machine case, namely, a structural result that upper bounds the number of times that large jobs are replaced, and a refined efficiency argument.
Similar to the general case, most of the difficulties arise when the last job has medium size. 
For concreteness, readers may consider the size of the last job $n$ as slightly larger than $\gamma = 0.38$, say, $p_n = 0.4$.

Recall that for the two-machine case we fix the parameters $\beta = \alpha = 0.2$.

\subsubsection{Bounding Major Replacements}

We are interested in jobs that have size at least $p_n$. 
We call such jobs \emph{major jobs} and refer to the replacements of major jobs as \emph{major replacements}.
In the case that there are only two machines, the number of major jobs is at most $4$ by a simple bin-packing argument.
(Recall that we focus on medium size last job, say $p_n = 0.4$.)
Further, only major jobs can replace major jobs. 
Hence, we can show sharp bound on the number of major replacements.
In particular, we show that when the last job is of medium size, there is either no major replacement, or at most one major replacement, depending on the size of the last job.
This is formulated as the following lemmas.

\begin{lemma}\label{lemma:two_machines_irreplaceable}
	If $\alg > 1.38$ and $p_n > \frac{1}{2+\alpha}$, then there is no major replacement.
\end{lemma}
\begin{proof}
	Recall that for the two machines case, we set $\beta = \alpha$.
	First we show that job $n$ cannot be replaced given $p_n > \frac{1}{2+\alpha}$.
	Observe that the replacer of $n$ must be of size $(1+\alpha)p_n > 1-\frac{1}{2+\alpha}$, which cannot be scheduled on the same machine with $n$ in the optimal schedule.
	Suppose $n$ is replaced (which can happen at most once) by $l$.
	Then we know that $n$ is pending from $r_l$ to $s_n$, where
	\begin{equation*}
		s_n - r_l >  (1.38-p_n) - (1-(1+\alpha)p_n) = 0.38 +\alpha p_n > \frac{1}{2+\alpha}.
	\end{equation*}
	
	Hence by Lemma~\ref{lemma:boundary_jobs}, we have $p_k,p_j > \frac{1}{2+\alpha}$, where $\{k,j\} = J(s_n)$ are the two jobs processed at time $s_n$.
	If $l\notin \{k,j\}$, then none of $k,j,n$ can be scheduled on the same machine with $l$ in $\opt$, which yields a contradiction;
	otherwise suppose $l = k$.
	Then we must have $J(r_l) = \{n,j\}$ as otherwise we also have the same contradiction as just argued.
	Then $n$ and $j$ must be scheduled on the same machine in $\opt$.
	If $r_n < r_j$, then $r_n < 1-p_n-p_j < 1-\frac{2}{2+\alpha} < \alpha p_n$.
	Hence $n$ must be scheduled at $r_n$, as otherwise by Lemma~\ref{lemma:all_jobs_large} the two jobs in $J(r_n)$ are of size larger than $\frac{p_n}{1+\alpha} > \frac{1}{3}$. Since
	\begin{equation*}
		r_l - r_n > s_j - r_n > (1.38-p_n-p_j) - (1-p_n-p_j) = 0.38 > \alpha,
	\end{equation*}
	it is impossible for $l$ to replace $n$ as $n$ has been processed a portion larger than $\alpha$.
	If $r_j < r_n$, then for the same reasoning $j$ is scheduled at $r_j$.
	As $s_j - r_j > (1.38-p_n-p_j) - (1-p_n-p_j) = 0.38$, we know that $j$ is replaced, which is impossible as by Lemma~\ref{lemma:reschedule_rule}, the job completed at $s_j$ is a job of size larger than $\frac{1}{2+\alpha}$, apart from $l,j$ and $n$.
	
	Hence we know that $n$ is never replaced. Now suppose some other job of size larger than $\frac{1}{2+\alpha}$ is replaced at time $r_x$.
	Then we have $p_x > 1-\frac{1}{2+\alpha}$, while the two jobs in $J(r_x)$ are of size larger than $\frac{1}{2+\alpha}$.
	Note that since $n$ is never replaced, it is not scheduled before $s_n$.
	Further by our assumption that no job arrives after $s_n$, we have $r_x < s_n$.
	Thus $n\notin J(r_x)$, which implies a contradiction.
\end{proof}

\begin{lemma}
	If $\alg > 1.38$ and $p_n\in (0.38,\frac{1}{2+\alpha}]$, then there is at most one major replacement.
\end{lemma}

The proof of the lemma is deferred to Section~\ref{ssec:two_machine_large}.
The main idea is that, when there are two or more major replacements, then we can find at least four major jobs.
Hence as long as we can find a job of size ``not too small'', then the bin-packing argument yields a contradiction.

\paragraph{Why is bounding the number of major replacements useful?}
Recall that in the efficiency argument, we upper bound the total waste by $\alpha m = 2 \alpha$ using the fact that each job can only create waste once at its release time if it replaces some other job, and the amount of waste created is at most $\alpha$ times the size of the replacer.
Suppose we can find one major job that does not replace any other job at its arrival.
Then, the upper bound of the total waste will significantly decrease by $\alpha p_n$.
(Recall that we focus on medium size last job, say $p_n = 0.4$.)

How do we find major jobs that do not replace other jobs at their arrival?
It turns out we can argue that in order to have $\alg > 1 + \gamma$, there must exists some major jobs whose final start times do not equal their release times.
For example, the last job $n$'s final start time definitely does not equal its release time.
For each of these jobs, if it replaces some other job at its arrival, it must be replaced later on in order to get rescheduled at its final start time.
Hence, a major replacement must occur for each of these jobs.
By bounding the number of major replacements, we get that some of these major jobs must not replace other jobs at their arrivals.

\subsubsection{Refined Efficiency Argument}

The second technical ingredient is a more careful efficiency argument.
Let $t$ be the last idle time before $s_n$. 
The total amount of work done in the optimal schedule after time $t$ is at most $2(1-t)$. 
(Recall that $\opt = 1$ and there are $m = 2$ machines.)

\paragraph{How much work does the algorithm process after time $t$?}

The algorithm is fully occupied from time $t$ to time $s_n$.
Further, let $\P$ be the total processing our algorithm does after $s_n$.
Then, the total amount of processing power after time $t$ by the algorithm is $2(s_n - t) + \P$.
However, some of the processing power is wasted (due to replacements) and some of the workload could have been done before time $t$ in $\opt$ (the leftover).
Recall that $\W_t$ is the amount of waste before time $t$.
Hence, the amount of waste located after time $t$ is $\W_1 - \W_t$.
Also recall that the amount of leftover is denoted as $\Delta_t$.
So we have: (the first inequality comes from Observation~\ref{observation:processing_after_t})
\[
2(1-t)\geq 2(s_n-t) + \P - \Delta_t -(\W_1 - \W_t) \geq 2(s_n-t) + \P - \frac{t}{2} -\W_1,
\]
where the second inequality follows by the leftover lemma.
Rearranging terms, the above implies:
\begin{equation*}\label{equation:efficiency_overview}
	\alg = s_n+p_n \leq 1+p_n-\frac{\P}{2} + \frac{t}{4} + \frac{\W_1}{2}.
\end{equation*}
Then, we will bound each of the terms $\P$, $t$, and $\W_1$:
$\P$ is trivially lower bounded by $p_n$;
$t$ is trivially upper bounded by $1$, while a more careful argument shows that $t \le 1 - p_n$ (Fact~\ref{fact:no_idle_after_1-p_n});
and $\W_1$ is trivially upper bounded by $2 \alpha$.
If we plug in the trivial bounds, we get:
\[
\alg \le 1.25 + \alpha + \frac{p_n}{2} ~.
\]
Hence, we recover the efficiency argument similar to what we have used in Section~\ref{sec:structural} (except that we further relax $\frac{m-1}{m} p_n \le p_n$ in the general case).
However, if we can obtain improved bounds on $\P$, $t$, or $\W_1$, we would have a better efficiency argument for upper bounding $\alg$.
It is usually impossible to get better bounds for all of these three quantities. 
Nonetheless, we manage to do so for at least one of them in all cases.

We have already provided an argument for getting a better upper bound of $\W_1$ by bounding the number of major replacements.
Next, we present some intuitions why it is possible to get improved bounds for $\P$ and $t$.

Since the jobs are scheduled greedily, if a replaced job $x$ is of size smaller than $p_n$, then it often happens that job $x$ is rescheduled after $s_n$. Hence while we suffer a loss in the total waste $\W_1$ due to the waste that comes from $x$, we have an extra gain of $p_x$ in $\P$.
In general, we will develop a unified upper bound on $\W_1 - \P$ (using the same spirit as in the proof of Lemma~\ref{lemma:bounded_waste} and~\ref{lemma:bounded_waste_weaker}), which measures the net waste due to replacements.

Apart from the trivial upper bound $1-p_n$ on $t$ (by Fact~\ref{fact:no_idle_after_1-p_n}), we can often derive better upper bounds on $t$ if we do not have a good upper bound on $\W_1 - \P$.
In the case when the total net waste is large, we can usually find many major jobs processed at the end of the schedule.
As we have only $m=2$ machines, during idle periods, only one job can be processed while no job is pending.
Suppose we find four major jobs processed at the end of the schedule, then at least three of them must be released after the last idle time $t$. Since $\opt = 1$, we must have $t \leq 1-2p_n$, which gives a better upper bound on $t$.

\subsection{Some Basic Facts for Scheduling on Two Machines} \label{ssec:some_basic_facts}

As before, let $n$ be the job completed last and we assume all release times, processing times, start times and completion times are distinct.
By the minimum counter-example assumption, no job arrives after time $s_n$.
Since $\alg =s_n + p_n> 1.38$, we have $s_n - r_n > (1.38-p_n) - (1-p_n) =0.38$.
Let $J(s_n) = \{k,j\}$ be the two jobs processed before $n$ is scheduled.
By Lemma~\ref{lemma:boundary_jobs}, we have $p_k,p_j > \min\{0.38, p_n\}$.

Let $t$ be the last idle time before $s_n$. By Fact~\ref{fact:no_idle_after_1-p_n}, we have $t\leq 1-p_n$.
Let $\P$ be the total processing our algorithm does after $s_n$.
By Observation~\ref{observation:processing_after_t}, we have (note that $m=2$)
\begin{equation*}
	2(1-t)\geq 2(s_n-t) + \P - \Delta_t -(\W_1 - \W_t) \geq 2(s_n-t) + \P - \frac{t}{2} -\W_1,
\end{equation*}
which implies
\begin{equation}\label{equation:efficiency}
	\alg = s_n+p_n \leq 1+p_n-\frac{\P}{2} + \frac{t}{4} + \frac{\W_1}{2}.
\end{equation}

\begin{definition}[Uncharged Jobs]
	We call a job $j$ \emph{uncharged}, if $j$ does not replace any job at $r_j$, or the job it replaces is rescheduled strictly after $s_n$.
\end{definition}

Note that if a job $j$ is uncharged, then its contribution to the RHS of (\ref{equation:efficiency}) is non-positive.
Let $\Q$ be the total size of uncharged jobs, and define $d:=1-p_n-t \geq 0$, we have
\begin{equation}
	\alg = s_n+p_n \leq 1+\frac{3p_n}{4} - \frac{\widetilde{\P}}{2} + \frac{1-d}{4} + \frac{\alpha}{2}(2-\Q)
	= 1.45 - (\frac{\Q}{10} + \frac{2\widetilde{\P} +d - 3p_n}{4}) \label{equation:extra},
\end{equation}
where $\widetilde{\P} = \sum_{j:s_j\in [s_n-p_js_n]}\max\{c_j - s_n,0\}$ is the total processing our algorithm does after $s_n$, excluding the jobs start strictly after $s_n$.
Hence as long as we can show that $\frac{\Q}{10} + \frac{2\widetilde{\P} +d - 3p_n}{4} \geq 0.07$, we have $\alg \leq 1.45-0.07 = 1.38$, contradicting our initial assumption.

Observe that since $\widetilde{\P} \geq p_n$, we always have
\begin{equation*}
	\frac{\Q}{10} + \frac{2\widetilde{\P} +d - 3p_n}{4} \geq \frac{\Q}{10} + \frac{d - p_n}{4}.
\end{equation*}

\begin{observation}[Bin-packing Constraint]\label{observation:bin-packing}
	It is impossible to find a set of jobs whose sizes cannot be packed into two bins of size $1$.
	For example, if we can find five jobs or size larger than $\frac{1}{3}$, or three jobs of size larger than $\frac{1}{2}$, then we have the contradiction that $\opt > 1$.
	Similarly, it is impossible to have three jobs of size at least $p > \frac{1}{3}$, and another job of size larger than $1-p$.
\end{observation}

\subsection{Upper Bounding Last Job} \label{ssec:two_machine_large}

We show that $p_n \leq \frac{1}{2+\alpha}$ in this section.
As it will be convenient for future analysis, we first rule out the case when $p_n > \frac{1+\alpha}{2}$.

\begin{lemma}
	We have $p_n \leq \frac{1+\alpha}{2}$.
\end{lemma}
\begin{proof}
	Assume the contrary that $p_n > \frac{1+\alpha}{2}$.
	We show that $\min\{p_k,p_j,p_n\} > \frac{1}{2}$, which contradicts $\opt = 1$ (Observation~\ref{observation:bin-packing}).
	
	Suppose $p_k = \min\{p_k,p_j,p_n\} \leq \frac{1}{2}$, then we have $s_k \leq r_n$.
	We show that $s_j > r_n$, which implies $p_j>p_n$.
	Suppose otherwise, then $k$ is the minimum job processed at time $r_n$.
	As $p_k < \frac{p_n}{1+\alpha}$ and $r_n - s_k\leq 0.5-0.38= 0.12 < \alpha p_n$, $k$ should have been replaced by $n$ (note that $n$ must be the largest pending job at $r_n$, if $s_j \leq r_n$).
	
	Hence we know that at $s_k$, both $n$ and $j$ are not released, which gives $r_k\leq s_k < \min\{r_n,r_j\}$.
	As $k$ must be scheduled with one of $j,n$ on the same machine in $\opt$, we have $s_k - r_k > (1.38-p_n-p_k) - (1-p_k -p_n) = 0.38$.
	Hence by Lemma~\ref{lemma:boundary_jobs}, the two jobs in $J(s_k)$ are of size larger than $\min\{s_k-r_k,p_k\} \geq 0.38$, which by Observation~\ref{observation:bin-packing} also contradicts $\opt = 1$.
\end{proof}

\begin{lemma}
	We have $p_n \leq \frac{1+\alpha}{2}$.
\end{lemma}
\begin{proof}
	Assume for contrary that $p_n > \frac{1}{2+\alpha}$.
	Recall from Lemma~\ref{lemma:two_machines_irreplaceable} that under this assumption, any job of size larger than $\frac{1}{2+\alpha}$ (including $n$) cannot be replaced.
	Hence $n$ is pending from $r_n$ to $s_n$, and is uncharged.
	Recall that we have $p_k, p_j > 0.38$ for $k,j\in J(s_n)$.
	We first show that if $n$ is scheduled with one of $k,j$ on the same machine in $\opt$, then $n$ is not the one released earlier.
	
	Suppose otherwise, then we have $s_n - r_n \geq (1.38-p_n) - (1-p_n-0.38) = 0.76$, which implies $\min\{p_k,p_j\} > \min\{ s_n-r_n, p_n \} > p_n$.
	Hence both $k$ and $j$ cannot be replaced.
	Moreover, we have $r_n < 1-2p_n$, which gives $d > p_n$.
	Suppose $s_k > s_j$, then $k$ is not charged: either $s_k \neq r_k$, or the job replaced by $k$ is rescheduled after $s_n$.	
	Then we have $\frac{\Q}{10} + \frac{d - p_n}{4} \geq 0.2p_n\geq 0.09$, which implies $\alg\leq 1.38$, a contradiction.
	
	We show that $\min\{p_k,p_j\} < \frac{1}{2+\alpha}$. Suppose otherwise, then $k,j,n$ are all of size lager than $\frac{1}{2+\alpha}$.
	Hence one of them, suppose $k$, is scheduled before one of $j,n$ on the same machine in $\opt$, which implies
	$r_k < 1-\frac{2}{2+\alpha}\leq \alpha p_k$.
	Since $s_k \neq r_k$, by Lemma~\ref{lemma:all_jobs_large}, the two jobs in $J(r_k)$ are both of size larger than $\frac{1}{3}$.
	Moreover, we have $j\notin J(r_k)$, as otherwise $p_j > s_n - r_k > \frac{2}{3}$.
	Hence there exist five jobs of size larger than $\frac{1}{3}$, contradicting $\opt = 1$.
	
	Next we show that $s_k \leq r_n$ and $s_j \leq r_n$.	
	Suppose $s_k > r_n$, then we have $p_k > p_n$ and $p_j < \frac{1}{2+\alpha}$.
	Hence we have $r_j \leq s_j < \min\{r_k,r_n\}$, which implies $s_j - r_j > (1.38-p_n-p_j) - (1-p_j-p_n) = 0.38$.
	Thus the two jobs in $J(s_j)$ (note that $k,j,n\notin J(s_n)$) are both of size larger than $0.38$, which contradicts $\opt = 1$.	
	Hence we know that both $k$ and $j$ starts before $r_n$.
	Moreover, we know that $n$ is the largest pending job at time $r_n$, which implies $\min\{p_k,p_j\} \geq \min\{0.38+\alpha p_n, \frac{p_n}{1+\alpha}\} = \frac{p_n}{1+\alpha}$.
	
	We proceed to show that $k,j$ are never replaced.
	
	Observe that by Lemma~\ref{lemma:reschedule_rule} at most one of $k,j$ is ever replaced.
	Suppose $k$ is replaced and $x$ is the job completed at $s_k$.
	Then as we already have four jobs ($x,k,j,n$) of size larger than $0.38$, $k$ can only be replaced by one of $x$ and $j$, while the other job is being processed while $k$ is replaced.
	Hence $k$ can only be replaced by $j$, as otherwise $j$ is of size $p_j > s_n - r_x > (s_n-r_n)+(c_x-r_x) > 0.38+(1+\alpha)p_n$, which cannot be scheduled on the same machine with any of $x,k,n$ in $\opt$.
	
	However, if $j$ replaces $k$ (at $r_j = s_j$), then we have $p_j > (1+\alpha)p_k > p_n$.
	First observe that $j$ cannot be schedule with $n$ on the same machine in $\opt$, as $s_j = r_j < r_n$, and $s_j+p_j+p_n > 1.38$. Hence $j$ must be scheduled with one of $k$ and $x$ on the same machine in $\opt$.
	
	\begin{figure}[h]
		\centering
		\vspace*{-10pt}
		\includegraphics[width = 0.65\textwidth]{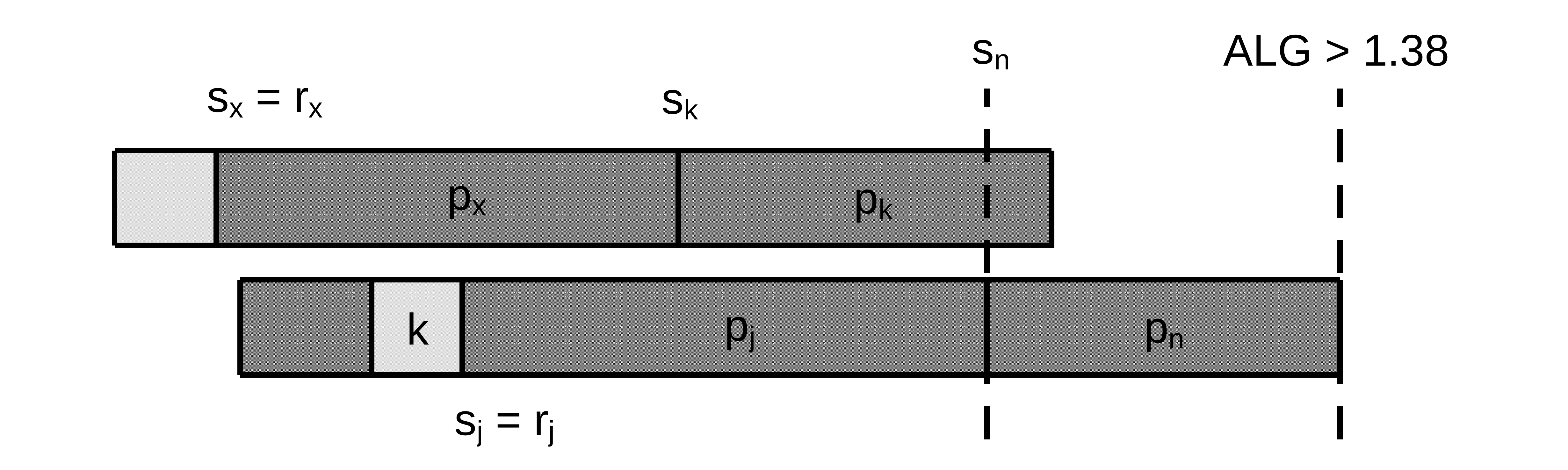}
		\vspace*{-10pt}
		\caption{Case when $j$ replaces $k$, when $x$ is being processed.}
		\label{fig:j_replace_k}
	\end{figure}
	
	Note that as we already have $p_x > p_k > 0.38$, and $p_j > p_n > \frac{1}{2+\alpha}$, any other job must have size strictly smaller than $1-0.38-\frac{1}{2+\alpha} = 0.166$.
	
	If $j$ is scheduled with $k$, then we have
	\begin{align*}
		& s_k(r_j) - r_k > s_j-\alpha p_j - r_k > (1.38-p_n-p_j) - \alpha p_j - (1-p_k-p_j)\\
		\geq & 0.38 - p_n - \alpha p_j + \frac{p_n}{1+\alpha} \geq 0.38 - 0.2\times (1-0.38) - (1-\frac{0.2}{1.2}\times 0.6) = 0.156,
	\end{align*}
	which implies that the job completed at $s_k(r_j)$ is of size at least $\min\{0.156+\alpha p_k, \frac{p_k}{1+\alpha}\} > 0.2$, as $k$ is the largest pending job at $r_k$, but not scheduled.
	Then we have a contradiction.
	
	If $j$ is scheduled with $x$ in $\opt$, then we have $r_x \leq 1-p_x-p_j$.
	As $k$ is scheduled with $n$ (and $k$ is released before $n$), we have $r_k \leq 1-p_k-p_n$.
	\begin{compactitem}
		\item If $s_x \neq r_x$, then $x$ is uncharged.
		As $s_k(r_j)-r_k > (1.38-p_n-p_j) - \alpha p_j - (1-p_k-p_n) > 0.38+p_k-(1+\alpha)p_n >0$, we know that $k$ is uncharged.
		Moreover, there is no idle time after $\max\{r_x,r_k\} < 1-p_n-0.38$, as at most one job is processed at idle time. Hence $d > 0.38$, which implies $\frac{\Q}{10} + \frac{d - p_n}{4} \geq (\frac{1}{10}(1+\frac{2}{1+\alpha}) - \frac{1}{4})p_n + \frac{0.38}{4} > 0.095$.
		\item If $s_x = r_x$ (refer to Figure~\ref{fig:j_replace_k}), then we have $p_k+ p_n \geq (s_n-c_x) + p_n > 1.38 - ((1-p_x-p_j) + p_x) = 0.38+ p_j$.
		Hence we have
		\begin{align*}
			s_k(r_j) - r_k & > (1.38 - p_j - p_n) - \alpha p_j - (1-0.38-p_j) \\
			& \geq 0.76 -p_n - \alpha p_j  \geq 0.76 - p_n - \alpha(1-\frac{p_n}{1+\alpha}) \\
			& \geq 0.76 - (1-\frac{\alpha}{1+\alpha})\frac{1+\alpha}{2+\alpha}-\alpha > 0.1,
		\end{align*}
		where we use $p_n \leq \frac{1+\alpha}{2+\alpha}$ as otherwise $p_k > \frac{1}{2+\alpha}$.
		Hence the job completed at $s_k(r_j)$ is of size larger than $\min\{0.1+\alpha p_k, \frac{p_k}{1+\alpha}\} > 0.176$, which is also a contradiction.
	\end{compactitem}
	
	Hence we conclude that none of $k,j$ is ever replaced. Assume $s_k < s_j$.
	
	If $j$ is scheduled on the same machine before one of $k,n$ in $\opt$, then $j$ is uncharged, and there is no idle time after $r_j \leq 1-\frac{2p_n}{1+\alpha}$.
	Hence we have $d = 1-p_n-t \geq \frac{1-\alpha}{1+\alpha}p_n$.
	We show that $k$ is also uncharged: if $k$ is charged, for the job $x$ replaced by $k$ must be rescheduled before $s_j$. Let $y$ be the job completed at $s_x$. Then we have $p_y > p_x > p_j > 0.38$, contradicting $\opt = 1$. Hence we have $\frac{\Q}{10} + \frac{d - p_n}{4} \geq \frac{1}{10}(1+\frac{2}{1+\alpha})p_n + \frac{1}{4}(\frac{1-\alpha}{1+\alpha}-1)p_n \geq 0.083$, which implies $\alg < 1.38$.
	
	Otherwise $k$ is scheduled before one of $j,n$ in $\opt$, and $k$ is pending from $r_k$ to $s_k$.
	Let $J(s_k) = \{x,y\}$. Note that $k,j,n\notin J(s_k)$.
	If $k$ is scheduled on the same machine with $n$ in $\opt$, then $s_k - r_k > (1.38-p_n-p_k)- (1-p_k-p_n) > 0.38$, which (by Lemma~\ref{lemma:boundary_jobs}) implies $p_x, p_y > 0.38$ and contradicts $\opt = 1$;
	if $k$ is scheduled with $j$, then $s_k - r_k > (1.38-p_n-p_k)- (1-p_k-\frac{p_n}{1+\alpha}) > 0.28$.
	Hence $p_x,p_y > 0.28 > \max\{ 1-2\times 0.38, \frac{1}{2}(1-p_n)\}$, which also contradicts $\opt = 1$.
\end{proof}

\subsection{Medium Last Job} \label{ssec:two_machine_boundary}

In this section, we show that if $p_n > 0.38$, then jobs of size larger than $0.38$ are ``almost irreplaceable'' (Lemma~\ref{lemma:major_replacement}).
We prove the following lemma, which is stronger than the one claimed in Section~\ref{sec:two-machine-overview}.
Different from Section~\ref{sec:two-machine-overview}, here we call a replacement \emph{major} if the job replaced is of size larger than $0.38$.

\begin{lemma}\label{lemma:major_replacement}
	If $p_n\in (0.38,\frac{1}{2+\alpha}]$, then there is at most one major replacement.
\end{lemma}
\begin{proof}
	Suppose there are two major replacements.
	
	Consider the first major replacement.
	Let $z$ be the replacer, $J(r_z) = \{x,y\}$ and $x$ be the job replaced.
	Then we have $p_y > p_x > 0.38$, and $p_z > (1+\alpha)0.38 = 0.456$.
	Let $l$ be the replacer in the second major replacement.
	Then we know that the job replaced by $l$ must be one of $x,y,z$.
	Observe that $n\in \{x,y\}$, since we have $\min\{p_z,p_l\} > 0.456 > \frac{1}{2+\alpha}$, and it is impossible to have five jobs of size larger than $0.38$.
	
	We first show that $l$ and $z$ cannot be scheduled on the same machine in $\opt$.
	
	Suppose otherwise, then we have $r_z + p_z + p_l \leq 1$ (as $r_z < r_l$).
	Hence the job replaced by $l$ is not $z$, as otherwise $p_z + p_l > ((1+\alpha)^2 + (1+\alpha))0.38 > 1$;
	the job replaced by $l$ is not $x$, as otherwise $y$ is completed at $s_x(r_l)$ and $n=x$, which gives $r_z + p_z + p_l > r_z + p_z + p_x \geq \alg > 1.38$.
	Thus $y$ is replaced by $l$, which gives $p_l \geq (1+\alpha)p_y$.
	However, since $r_z + p_z + p_x + p_y > 1.38$, we have
	\begin{align*}
		r_z + p_z + p_l &> (r_z + p_z + p_x + p_y) + \alpha p_y - p_x
		> 1.38 + \alpha\cdot 0.38 - \frac{1-(1+\alpha)\cdot 0.38}{1+\alpha} > 1,
	\end{align*}
	which is also a contradiction.
	
	As $l$ and $z$ cannot be scheduled on the same machine in $\opt$, we know that apart from $x,y,z$ and $l$, any other job must be of size less than $1-(2+\alpha)\cdot0.38 = 0.164 < \frac{0.38}{1+\alpha}$.
	Moreover, as $x$ and $y$ are released before $z$ and $l$, we have $r_x \leq 1-p_x-\min\{p_z,p_l\}$ and $r_y \leq 1-p_y-\min\{p_z,p_l\}$.
	
	Depending on which job is replaced by $l$, we divide our analysis into three cases.
	
	\paragraph{Case-1: $x$ is replaced.}
	If $x$ is replaced by $l$, then we know that $y$ is completed at $s_x(r_l)$.
	Hence $n=x$, which is scheduled after $z$ or $l$.
	As $r_z - s_n(r_z) \leq \alpha p_z$, we have
	\begin{align*}
		s_n(r_z) - r_n & > (1.38-p_n-p_z) -\alpha p_z - (1-p_n-\min\{p_z, p_l\}) \\
		& = 0.38+\min\{p_z, p_l\}-(1+\alpha)p_z > 2.2\times 0.38 - 1.2\times 0.62 = 0.092.
	\end{align*}
	
	\vspace*{-10pt}
	\begin{figure}[H]
		\centering
		\includegraphics[width = 0.65\textwidth]{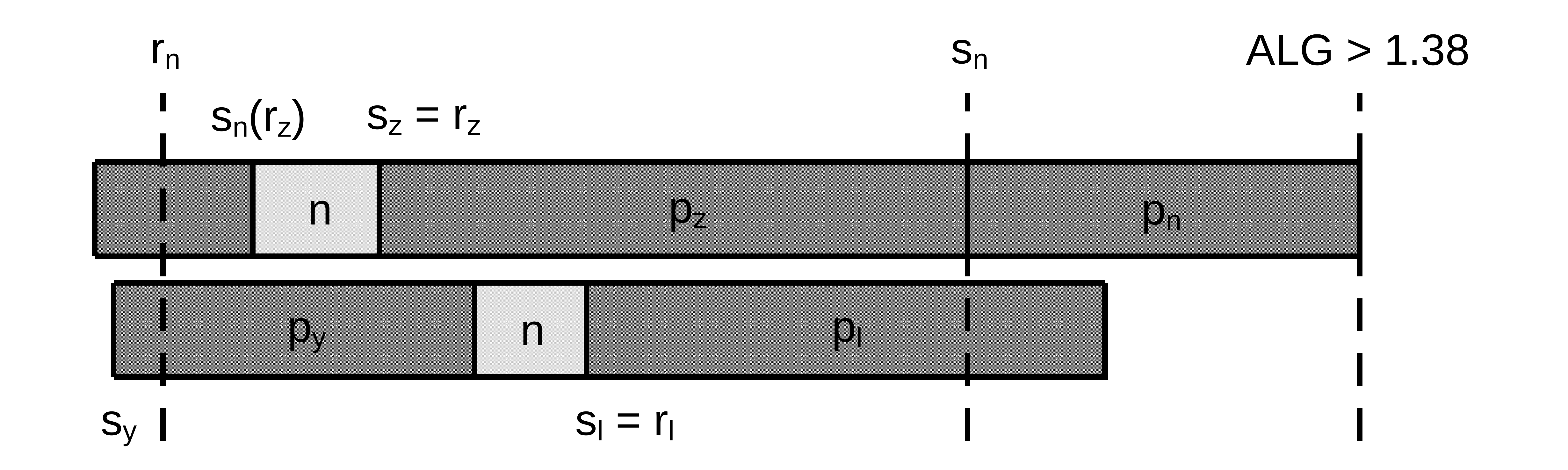}
		\vspace*{-5pt}
		\caption{Case when $n=x$ is replaced by $l$.}
		\label{fig:l_replace_x}
	\end{figure}
	Hence we know that $n$ is pending from $r_n$ to $s_n(r_z)$, which means that the job completed at $s_n(r_z)$ is of size at least $\min\{ 0.092+\alpha p_n, \frac{p_n}{1+\alpha} \} \geq 0.168$, contradicting $\opt = 1$.
	
	\paragraph{Case-2: $y$ is replaced.}
	First note that $n\neq y$, as otherwise (similar to \textbf{Case-1}) we have
	\begin{align*}
		s_n(r_l) - r_n  &> 0.38+\min\{p_z, p_l\} - (1+\alpha) \max\{p_z,p_l\}
		> 2.2\times 0.38 - 1.2\times 0.62 = 0.092,
	\end{align*}
	which implies a contradiction.
	Thus we have $n=x$. Consider which job is completed at $s_n$:
	\begin{compactitem}
		\item if it is $z$, then by \textbf{Case-1} we have $s_n(r_z) - r_n  > 0.092$;
		\item if it is $l$, then $s_y(r_l) - r_y  > (1.38-p_x-p_l)-\alpha p_l-(1-p_y-\min\{p_z,p_l\}) > 0.092$;
		\item if it is $y$, then we know that the job processed on the other machine at $s_n$ is either $z$ or $l$.
		If it is $z$, then we have $s_n(r_z) > 1.38 - p_n-(1+\alpha)p_z$; otherwise we have $s_y(r_l) > 1.38-p_n-(1+\alpha)p_l$. Thus we go to one of the above two cases.
	\end{compactitem}
	Hence in all cases, we can find a job completed at either $s_n(r_z)$ or $s_y(r_l)$ that is of size larger than $0.168$, which contradicts $\opt = 1$.
	
	\paragraph{Case-3: $z$ is replaced.}
	If $z$ is replaced, then we know that $y$ is not replaced (thus $n=x$).
	Moreover, $y$ must be processed at $r_l$, which gives $p_y > p_z$.
	Observe that $p_l \geq (1+\alpha)p_z > 1-p_z > 1-p_y$. Hence in $\opt$, $l$ cannot be scheduled with $z$ or $y$.
	Then we have $p_y \leq 1-p_z < p_l$, which implies $c_y < c_l$.
	Hence we know that $z$ is rescheduled at $c_y$.
	Then again, we have
	\begin{align*}
		s_n(r_z) - r_n & > (1.38-p_n-p_l) -\alpha (p_l+p_z) - (1-p_n-p_l) \\
		& = 0.38 - \alpha(p_l+p_z) \leq 0.38-0.2\times (0.62 + 0.5) = 0.156,
	\end{align*}
	which implies that the job completed at $s_n(r_z)$ is of size at least $0.232$, contradicting $\opt = 1$.
\end{proof}

Now with the help of Lemma~\ref{lemma:major_replacement}, we show that we can push the upper bound of $p_n$ from $\frac{1}{2+\alpha}$ to $0.38$.
Depending on whether $n$ is ever replaced, we use different proof strategies.

\begin{lemma}
	If $0.38 < p_n \leq \frac{1}{2+\alpha}$ and $n$ is never replaced, then $\alg \leq 1.38$.
\end{lemma}
\begin{proof}
	Note that in this case $n$ is uncharged. Recall that we have $p_k>0.38$ and $p_j > 0.38$.
	
	If $n$ is scheduled before one of $k,j$ in $\opt$, then we have $r_n \leq 1-p_n-0.38$, which implies $s_n - r_n > 0.76$.
	Hence we have $p_k>p_n$, $p_j > p_n$ and last idle time $t < r_n < 1-2p_n$.
	Suppose $s_k > s_j$.
	As it is impossible to have $p_k>0.76$ and $p_j>0.76$, we have $s_k > r_n$.
	We show that $k$ is uncharged: if the job replaced by $k$ (at $r_k$) is of size less than $p_n$, then it cannot be rescheduled before $s_n$; otherwise it is a major replacement, which (by Lemma~\ref{lemma:major_replacement}) implies that $k$ is not replaced. Hence we have $r_k = s_k$, and the job replaced by $k$ is rescheduled before $s_n$, which is impossible, as $j$ is processed from $s_k$ to $s_n$.
	
	Then we have $\frac{\Q}{10}+\frac{d-p_n}{4} > 0.2p_n > 0.076$, which gives $\alg \leq 1.38$.
	
	Hence two of $k,j,n$ are scheduled together in $\opt$, while $n$ is not the one released earlier.
	
	As before, we show that none of $k,j$ is ever replaced.
	
	Suppose $k$ is replaced. Let $x$ be the job completed at $s_k$. We have $p_x > p_k$, and we know that $k$ must be replaced by one of $x$ and $j$, while the other job is being processed when $k$ is replaced.
	If $k$ is replaced by $x$, then either $k$ is uncharged (if $s_k(r_x) > s_j$), or $j$ is uncharged (if $s_k(r_x) < s_j$), because the job replaced must be rescheduled after $s_n$.
	Moreover, we have $c_k - s_n > p_x+p_k-p_j > (2+\alpha)0.38 - 0.62 > 0.216$, which gives
	\begin{align*}
		\frac{\Q}{10}+\frac{2\widetilde{\P} -3p_n}{4} &\geq \frac{1}{10}(p_n+0.38)+\frac{1}{4}(2\times 0.216 - p_n)
		\geq \frac{0.38}{10}+\frac{0.216}{2}-\frac{0.15}{2.2} > 0.0778.
	\end{align*}
	
	If $k$ is replaced by $j$, then we have $p_j > (1+\alpha)p_k > 0.456>p_n$.
	Note that there is no idle time after $\max\{r_k,r_x\} < 1-0.38-p_n$, which gives $d > 0.38$.
	
	If $s_k(r_j) < s_x$, then $x$ is uncharged, and $c_k - s_n > p_x+p_k-(1+\alpha)p_j > 0.016$, which implies
	\begin{align*}
		\frac{\Q}{10}+\frac{2\widetilde{\P} +d -3p_n}{4} &\geq \frac{1}{10}(p_n+0.38)+\frac{1}{4}(2\times 0.016 +0.38 - p_n) \\
		&\geq \frac{0.38}{10}+\frac{0.412}{4}-\frac{0.15}{2.2} > 0.0728.
	\end{align*}
	
	Hence we have $s_k(r_j) > s_x$, which means that $k$ is uncharged.
	Observe that if $x$ is also uncharged then we are done as
	\begin{align*}
		\frac{\Q}{10}+\frac{d -p_n}{4} &\geq \frac{1}{10}(p_n+0.38+0.38)+\frac{1}{4}(0.38 - p_n) 
		\geq \frac{0.38}{5}+\frac{0.38}{4}-\frac{0.15}{2.2} > 0.1.
	\end{align*}
	
	Hence we have $r_x = s_x < s_k(r_j)$, and the job replaced by $x$ is completed during $(s_x, s_k(r_j))$.
	Observe that we have $p_x > p_k > 0.38$ and $p_j > p_n > 0.38$, hence any other job must be of size less than $0.24$.
	Hence we have $r_k > s_x$, as otherwise the job replaced by $x$ (which is of size less than $0.24$) will be rescheduled after $s_n$.
	
	Since $r_k < 1-p_k-p_n$, we have $s_k(r_j) - r_k > (1.38-p_n-p_j) \alpha p_j - (1-p_k-p_n) = 0.38+p_k - (1+\alpha)p_j > 0$, which means that $k$ is not scheduled at $r_k$.
	Hence the job $y$ (apart from $x$) processed at $r_k$ is processed at least $\alpha p_k$ (as $0.24 < \frac{p_k}{1+\alpha}$).
	Moreover, we have $s_y(r_k) > s_x$, as otherwise $x$ is uncharged.
	Hence we have $t < s_y(r_k) < r_k - \alpha p_k < 1-p_n - 0.456$, which gives
	\begin{align*}
		\frac{\Q}{10}+\frac{d -p_n}{4} &\geq \frac{1}{10}(p_n+0.38)+\frac{1}{4}(0.456 - p_n)
		\geq \frac{0.38}{10}+\frac{0.456}{4}-\frac{0.15}{2.2} > 0.083.
	\end{align*}
	
	Hence we can assume that none of $k,j$ has been replaced. Assume $s_k < s_j$ ($j$ is uncharged).
	
	If $j$ is scheduled before one of $k,n$ on the same machine in $\opt$, then there is no idle time after $r_j \leq 1-2\times 0.38=0.24$, and $k$ is also not charged: if $k$ is charged, then for the job $x$ replaced by $k$ and the job $y$ completed at $s_x$ we have $p_y > p_x > p_j > 0.38$, contradicting $\opt = 1$. Hence
	\begin{align*}
		\frac{\Q}{10}+\frac{d -p_n}{4} &\geq \frac{1}{10}(p_n+0.76)+\frac{1}{4}(0.76 - 2p_n) \geq \frac{0.76}{10}+\frac{0.76}{4}-\frac{0.4}{2.2} > 0.084.
	\end{align*}
	
	Otherwise $k$ is scheduled before one of $j,n$ in $\opt$. Observe that $s_k-r_k>0$.
	Let $J(s_k) = \{x,y\}$. Note that $k,j,n\notin J(s_k)$.
	If $k$ is scheduled with $n$, then $s_k - r_k > (1.38-p_n-p_k)- (1-p_k-p_n) > 0.38$, which (by Lemma~\ref{lemma:boundary_jobs}) implies $\min\{p_x, p_y\} > 0.38$ and contradicts $\opt = 1$;
	if $k$ is scheduled with $j$, then $s_k - r_k > (1.38-p_n-p_k)- (1-p_k-0.38) = 0.76-p_n$.
	Hence $\min\{p_x, p_y\} > \min\{ 1-2\times 0.38, \frac{1}{2}(1-p_n)\}$, which also contradicts $\opt = 1$.
\end{proof}

\begin{lemma}
	If $0.38 < p_n \leq \frac{1}{2+\alpha}$ and $n$ is ever replaced, then $\alg \leq 1.38$.
\end{lemma}
\begin{proof}
	Observe that since $p_n > 0.38$ and $n$ is ever replaced, by Lemma~\ref{lemma:major_replacement} we know that $n$ is replaced exactly once, and any other job of size larger than $0.38$ is never replaced.
	
	Suppose $n$ is replaced by $l$.
	Let $J(r_l) = \{n,x\}$. Then we must have either $l\in J(s_n)$ or $x\in J(s_n)$, as otherwise we have five jobs of size larger than $0.38$.
	
	We first consider the case when $J(s_n) = \{l,x\}$.
	Note that $l,x$ are never replaced.
	
	\begin{enumerate}
		\item If $s_x > s_n(r_l)$, then we have $r_x > s_n(r_l)$, as $x$ cannot be pending at $s_n(r_l)$.
		Consider the instance with jobs released after $s_n(r_l)$ removed.
		Let $\opt'$ be the new optimal makespan and $\alg'$ be the makespan of our algorithm on the new instance.
		We have $\opt' \leq 1-\min\{p_x,p_l\}$ (as $x,l$ are released after $s_n(r_l)$ and cannot be scheduled on the same machine in $\opt$), while $\alg' = s_n(r_l)+p_n = \alg - (s_n - s_n(r_l)) \geq \alg - 1.38\cdot \min\{p_x,p_l\}$, which gives a smaller counter-example (the inequality holds since $s_n-r_l\leq \min\{p_x,p_l\}$ and $r_l - s_n(r_l) \leq \alpha p_l \leq \alpha (p_x+p_n-0.38) < 0.38 p_x$).
		
		\item If $r_x \leq s_x < s_n(r_l)$, then we compare the release times of $n$ and $x$.
		Observe that $r_l > \max\{r_n, r_x\}$.
		If $r_n<r_x$, then $t < r_n < 1-2p_n$.
		We show that $x$ is uncharged: any job $y$ replaced by $x$ at $r_x < r_n$ cannot be scheduled before $s_n$, as there is at most one major replacement.
		As $s_n(r_l) > s_x > r_n$, $n$ is uncharged, which implies $\frac{\Q}{10}+\frac{d-p_n}{4} > 0.2\times 0.38 = 0.076$.
		If $r_x<r_n$, then $n$ is uncharged, as any job replaced by $n$ can only be scheduled after one of $x,n$ is completed.
		\begin{compactitem}
			\item If $x$ is scheduled together with one of $n,l$ in $\opt$, we have $s_x - r_x > (1.38-p_x-p_n) - (1-p_x-p_n) = 0.38$, which means that the two jobs processed at $s_x$ (note that $n,x,l\notin J(s_x)$) are of size larger than $0.38$, contradicting $\opt = 1$;
			\item otherwise we have $r_x < r_n < 1-p_n-p_l$, which implies $d > p_l$. If $x$ is uncharged then we have $\frac{\Q}{10}+\frac{d-p_n}{4} > 0.095$; otherwise we have $p_x > 1.38-p_n-r_x > 0.38+p_l$. Moreover, we have $s_n(r_l)-r_n > (1.38-p_n-p_l)-\alpha p_l - (1-p_n-p_l) > 0.092$, which means that the job completed at $s_n(r_l)$ is of size larger than $0.168$. Then we have a contradiction as $\min\{0.168+p_x,0.168+p_n+p_l \} > 0.168+(2+\alpha)0.38 > 1$.
		\end{compactitem}
	\end{enumerate}
	
	Hence we have $J(s_n) \neq \{l,x\}$, which means that at least three of the four jobs in $\{n,x,l\} \cup J(s_n)$ are released after $t$ (as $t < r_l$), which implies $t < 1-2p_n$ and $d > p_n$.
	
	We first consider the case when $l\notin J(s_n)$ and $x\in J(s_n)$. 
	Suppose $x = j$, then we have $c_k - s_n > p_l+p_k-p_x > (2+\alpha)0.38 - 0.62 > 0.216$.
	Hence we have
	\begin{equation*}
		\frac{\Q}{10}+\frac{2\widetilde{\P} + d - 3p_n}{4} \geq \frac{p_n}{10}+\frac{2\times 0.216}{4} > 0.146.
	\end{equation*}
	
	Next we consider the case when $l\in J(s_n)$ and $x\notin J(s_n)$.
	Suppose $l=k$.
	Observe that at any time from $\max\{s_n(r_l),s_x\}$ to $s_n$, the minimum job being processed is of size at least $p_n$.
	Hence any job replaced at or after $\max\{s_n(r_l),s_x\}$ must be rescheduled after $s_n$.
	Thus job $j$ and one of $n,x$ are uncharged.
	Then we have $\frac{\Q}{10}+\frac{d - p_n}{4} \geq \frac{p_n}{5} > 0.2\times 0.38 = 0.076$.
\end{proof}

\subsection{Refined Efficiency Argument} \label{ssec:two_machine_small}

Observe that the upper bound (\ref{equation:extra}) on $\alg$ is quite loose when $p_n$ and $d$ are very small.
Actually, since there are only two machines, if time $t$ is idle, then there is one job being processed (it is impossible to have two idle machines in the minimum counter-example).
Hence we should have a better upper bound on $\Delta_t$, compared to Lemma~\ref{lemma:leftover} (The Leftover Lemma).

Let $t$ be the last idle time before $s_n$, and let $J(t) = \{i\}$.
Let $\opt_t$ be the makespan of the optimal schedule of the jobs released before $t$.
Then we have $c_i(t) := s_i(t)+p_i \leq 1.38\cdot\opt_t$, as otherwise we can remove all jobs released after $t$, and obtain a smaller counter-example.
Define $p := \max\{\opt_t - t,0\}$. Then we have $p\leq \opt_t\leq t+p$.
Note that by optimality of $\opt_t$, the total processing $\opt$ does after time $t$ is at least $p+\sum_{j: r_j\geq t}p_j$.
Hence the total size of jobs released after time $t$ is $\sum_{j: r_j\geq t}p_j \leq 2(1-t)-p$.

\begin{claim}
	There exists a job of size at least $p_n$ released after $t$ that does not replace other jobs at its release time.
\end{claim}
\begin{proof}
	Suppose otherwise, then every job of size at least $p_n$ released after $t$ must be scheduled immediately (but can possibly be replaced later).
	If there is any job of size $p_x < p_n$ released after $t$, then we argue that the instance with $x$ removed is a smaller counter-example: the schedule produced by our algorithm on the new instance is identical to the original scheduled (projected on jobs of size at least $p_n$), as the behavior of every job of size at least $p_n$ is unchanged.
	
	Hence in the minimum counter-example, all jobs released after $t$ are of size at least $p_n$.
	Then the first job released after $t$ does not replace any job, as there is an idle machine.
\end{proof}

By Observation~\ref{observation:processing_after_t} we have (recall that $\P$ is the total processing our algorithm does after $s_n$)
\begin{equation}\label{equation:W-P}
	2(1-t) \geq 2(s_n - t)+\P-\Delta_t(t) - (\W_1 - \W_t(t)).
\end{equation}
Note that $\W_1 - \W_t(t)$ is the total waste created by jobs released after time $t$, which is at most $\alpha(2-2t-p-p_n)$.
Rearranging the inequality and by $\P \geq p_n$, we have
\begin{equation*}
	\alg = s_n+p_n \leq 1.2 + 0.4p_n + 0.5\Delta_t(t) - 0.2t - 0.1p.
\end{equation*}

Note that we have $\Delta_t(t) \leq c_i(t) - t - p\leq 0.38\cdot \opt_t\leq 0.38(t+p)$.
Applying the upper bound on $\Delta_t(t)$, we have
\begin{equation}\label{equation:t_and_p}
	\alg = s_n+p_n \leq 1.2 + 0.4p_n + 0.09(t+p) - 0.1t.
\end{equation}

Observe that we have $t+p\leq 1$ and the following lower bound on $t$ (from (\ref{equation:efficiency})):
\begin{equation*}
	1.38 < \alg \leq 1+\frac{p_n}{2}+\frac{t}{4}+\frac{\alpha}{2}(2-p_n),
\end{equation*}
which implies $t > 4(0.38-\frac{p_n}{2}-\frac{\alpha}{2}(2-p_n)) = 0.72-1.6p_n$.
Hence we have
\begin{align*}
	\alg &= s_n+p_n \leq 1.2 + 0.4p_n + 0.09 - 0.1(0.72-1.6p_n) =1.218+0.56p_n.
\end{align*}

Thus immediately we can show that $p_n$ cannot be too small, as otherwise we have the contradiction that $\alg \leq 1.218+0.56\times 0.28 = 1.3748$.

\begin{lemma}[Lower Bound on $p_n$]
	We have $p_n > 0.28$.
\end{lemma}

We show the following lemma, which will be the main framework towards deriving a contradiction, given that $0.28 < p_n \leq 0.38$.

\begin{lemma}\label{lemma:p_n_leq_0.38}
	Given that $p_n\leq 0.38$, if we have $\W_1 - \W_t(t) - (\P-p_n) \leq \alpha(2-2t-p-2p_n)$, then we can show that $\alg \leq 1.38$.
\end{lemma}
\begin{proof}
	Applying the upper bounds on $\W_1 - \W_t(t) - (\P-p_n)$ and $\Delta_t(t)$ to (\ref{equation:W-P}), we obtain the following stronger version of (\ref{equation:t_and_p}).
	\begin{align*}
		\alg & = s_n+p_n \leq \frac{1}{2}(2-p_n+0.38(t+p)+\alpha(2-2t-p-2p_n))+p_n \\
		& = 1.2 + 0.3p_n + 0.09(t+p) -0.1t.
	\end{align*}
	
	Since $t+p\leq 1$ (we do not use the upper bound given by Lemma~\ref{lemma:upper_bound_opt_t}), we have
	\begin{align*}
		\alg & = s_n+p_n \leq 1.2 + 0.3p_n + 0.09 - 0.1(0.72-1.2p_n) \\
		& =1.218+0.42p_n\leq 1.218+0.42\times 0.38\leq 1.3776,
	\end{align*}
	where in the first inequality we use a stronger lower bound on $t$: $t > 4(0.38-\frac{p_n}{2}-\frac{\alpha}{2}(2-2p_n)) = 0.72-1.2p_n$, which holds only when $\W_1 - \W_t(t) - (\P-p_n) \leq \alpha(2-2t-p-2p_n)$.
\end{proof}

Notice that the ``if'' condition of Lemma~\ref{lemma:p_n_leq_0.38} holds if there exists two jobs of size at least $p_n$ released after $t$ that are uncharged.

Next we show an upper bound on $\opt_t$.
Note that if $p\neq 0$, then we have $t+p=\opt_t$.
Hence the upper bound holds for $t+p$ when $p\neq 0$.

\begin{lemma}\label{lemma:upper_bound_opt_t}
	We have $\opt_t \leq 1-p_n$.
\end{lemma}
\begin{proof}
	Suppose $\opt_t > 1-p_n$, then we know that any job of size at least $p_n$ released after time $t$ must be scheduled on the same machine in $\opt$ (otherwise $\opt_t$ is not optimal).
	Also by optimality of $\opt_t$, we have $c_i(t) \geq \opt_t$.
	
	Recall that we have $\min\{p_k,p_j,p_n\}\geq p_n$.
	Since only one job is processed at time $t$, we know that at least two of $k,j,n$ are released after $n$.
	Suppose $l\neq n$ is the largest job released after time $t$, then we have $p_i > c_i(t) - t > (1-p_n) - (1-p_n-p_l) = p_l$, which means that $p_i$ cannot be replaced after $t$, and $c_i(t) = c_i$.
	Observe that if there are two jobs of size at least $p_n$ released after $t$ that are uncharged, then by Lemma~\ref{lemma:p_n_leq_0.38} we have $\alg\leq 1.38$.
	
	Next we prove the existence of uncharged jobs.
	
	If $i\in J(s_n)$, suppose $i=k$, then we have $c_i = c_k > s_n > 1$.
	Note that $n$ is uncharged: any job replaced by $n$ must be rescheduled after $s_n$.
	We show that $j$ is also uncharged.
	Suppose otherwise, then the job replaced by $j$ must be $n$, and $j$ must also be replaced, as $c_j > 1$.
	Hence we have $p_j \geq (1+\alpha)p_n$ and the replacer of $j$ is of size at least $(1+\alpha)^2 p_n$, which is impossible, as $p_n+(1+\alpha)p_n + (1+\alpha)^2 p_n > 0.28\times (1+1.2+1.44) > 1$.
	
	If $i\notin J(s_n)$, then $k,j,n$ are all released after time $t$.
	Note that it is impossible to have four jobs of size at least $p_n$ released after $t$, as $4p_n > 1$.
	Observe that none of $k,j$ or $n$ can be replaced.
	Suppose otherwise, let $l$ be the replacer and $M$ be the machine where the replacement happens.
	Then the first job completed on $M$ after $r_l$ is a job of size at least $(1+\alpha)p_n$ that is not $k,j$ or $n$.
	Hence all of $k,j$ and $n$ are never replaced, thus uncharged (as $s_n > 1$).
\end{proof}

Lemma~\ref{lemma:upper_bound_opt_t} helps us to improve the lower bound on $p_n$.

\begin{corollary}[Improved Lower Bound]
	We have $p_n > \frac{1}{3}$.
\end{corollary}
\begin{proof}
	Assume for contrary that $p_n \leq \frac{1}{3}$.
	If $p = 0$, then by (\ref{equation:t_and_p}), we have $\alg\leq 1.2+0.4p_n \leq 1.34$;
	otherwise by Lemma~\ref{lemma:upper_bound_opt_t} we have $t+p=\opt_t \leq 1-p_n$.
	Hence we have
	\begin{align*}
		\alg & = s_n+p_n \leq 1.2 + 0.4p_n + 0.09(1-p_n) - 0.1(0.72-1.6p_n) \\
		& =1.218+0.47p_n\leq 1.218+0.47\times \frac{1}{3} \leq 1.3747,
	\end{align*}
	where in the first inequality we use $t > 0.72-1.6p_n$.
\end{proof}	

It remains to prove the following lemma. Recall that so far we have shown that $p_n\in(\frac{1}{3},0.38]$.

\begin{lemma}\label{lemma:two_uncharged_after_t}
	We have $\W_1 - \W_t(t) - (\P-p_n) \leq \alpha(2-2t-p-2p_n)$, given that $p_n\in(\frac{1}{3},0.38]$.
\end{lemma}
\begin{proof}
	If $n$ is never replaced, i.e., is pending from $r_n$ to $s_n$, then $n$ is uncharged.
	Note that $n\neq i$. Hence $n$ is released after $t$.
	We show that at least one of $k,j$ is uncharged.
	Note that it is impossible that both $k,j$ are replaced.
	Hence one of them, suppose $k$, is not replaced, thus uncharged.
	Moreover, $k$ must be released after $t$: otherwise $k=i$, and we have the contradiction that $c_i(t) \geq s_n = \alg - p_n > 1.38(1-p_n) \geq 1.38\cdot \opt_t$.
	
	Otherwise let $r_l$ be the last time $n$ is replaced, and $J(r_l) = \{n,x\}$. Then we have $p_x > p_n$.
	
	We show that both $k,j$ are never replaced.
	Suppose the contrary that $k$ is replaced. Then $k$ cannot be replaced by $n,l$ or $j$ (otherwise $j$ must also be replaced).
	As there cannot exist five jobs of size larger than $\frac{1}{3}$, we must have $k=l$ and $x=j$, i.e., $k$ replaces $n$ while $j$ is being processed, and then $k$ is replaced by some job $y$.
	If $c_k - s_n \geq 2\alpha p_n$, then we already have $\W_1 - \W_t(t) - (\P-p_n) \leq \alpha(2-2t-p-2p_n)$;
	otherwise $p_j > p_y+p_k-(c_k-s_n) \geq (2+\alpha)(1+\alpha)p_n - 2\alpha p_n > \frac{2}{3}$, which also contradicts $\opt =1$.
	Hence both $k,j$ are not charged. Moreover, for the same reason argued above, both $k,j$ are released after $t$.
\end{proof}

\section{Other Candidate Algorithms} \label{sec:candidate_alg}

All the candidate algorithms are based on LPT. That is, whenever there is an idle machine, we always schedule the largest job.
The only difference is the replacement rule. We will show that none of the them can beat the ratio of $1.5$.

\paragraph{Candidate Algorithm 1.}
Fix any constant $0<\rho<1$. Upon the arrival of a job $j$, job $k$ can be replaced by job $j$ if $k$ is the smallest processing job, $p_k<p_j$ and job $k$ has been processed no larger than $\rho$ fraction.

\noindent
\textbf{Counter example.}
At $t=0$, $m$ identical jobs come with $p_1 = p_2 = \cdots = p_m = 1$. Each of them is scheduled on a machine. At $t=\rho$, job $(m+1)$ comes with $p_{m+1}=1+\xi$ ($\xi$ is an infinitesimal amount). Then one of jobs $1$ to $m$ is replaced by job $(m+1)$. At $t=2\rho$, job $(m+2)$ comes with $p_{m+2}=1+2\xi$, then job $(m+1)$ is replaced by job $(m+2)$. The same thing goes on and on, and at time $t=1$, job $(m - 1 + \lceil \frac{1}{\rho} \rceil)$ is replaced by job $(m + \lceil \frac{1}{\rho} \rceil)$. After then, at $t=1$, $(m-\lceil \frac{1}{\rho} \rceil)$ jobs come with $p_{m + \lceil \frac{1}{\rho} \rceil + 1} = p_{m + \lceil \frac{1}{\rho} \rceil + 2} = \cdots = p_{2m} = 1$. Then there are $m$ pending jobs but only $(m-1)$ idle machines, so $\alg = 3$. In the optimal solution, no replacement happens, and $\opt = 2 + \lceil \frac{1}{\rho} \rceil\xi$.

\paragraph{Candidate Algorithm 2.}
Fix $0<\rho<1$ and $1 < \mu < 2$. Upon the arrival of a job $j$, job $k$ can be replaced by job $j$ if $\mu p_k<p_j$ and job $k$ has been processed no larger than $\rho$ fraction.

We show a counter example using $\mu = 3/2$ and $\rho = 1/2$. This counter example can be generalized to any $\mu$ and $\rho$ such that $\mu + \rho \le 2$.

\noindent
\textbf{Counter example.}
At $t=0$, $m$ jobs come, with $p_1 = m + 1, p_2 = m + 2, \cdots, p_m = 2m$. At time $t=m$, job $(m+1)$ comes with $p_{m+1}=3m$. Then, as job $m$ is the only job that is processed at most half, job $m$ is replaced by job $(m+1)$. After the replacement, $(m-1)$ jobs come with $p_{m+2} = 2m + 1, p_{m+3} = 2m + 2, \cdots, p_{2m} = 3m - 1$. Then $\alg = 6m$, while in the optimal solution, no replacement happens, thus $\opt = 4m + 1$.

\paragraph{Candidate Algorithm 3.}
Fix a target performance ratio $1+\gamma$. When a job $j$ comes, schedule it virtually to $\alg$, and calculate the current optimal solution with all the jobs that have been released. If the ratio can still be bounded in $1+\gamma$, do not replace any jobs; otherwise choose one job to replace.

\noindent
\textbf{Counter example.}
At $t=0$, $m$ jobs comes first, with $p_1=2m, p_2=2m+1, \cdots, p_m = 3m - 1$. After each of them has been scheduled on a machine, at $t=0$, another $m$ jobs come, with $p_{m+1} = 3m, p_{m+2} = 3m + 1, \cdots, p_{2m} = 4m-1$. At this time, since the local $\alg$ and local $\opt$ are exactly the same, no replacement happens. Then at $t=3m-1$, another job comes with $p_{2m+1} = 3m$. So for this instance, $\alg = 6m-1$; while in the optimal solution, three smallest jobs (jobs $1$, $2$, $3$) are scheduled on the same machine, while all other jobs are paired up with the smallest with largest, and $\opt=4m+3$.

\medskip

For our algorithm \lptr, some may wonder what happens if $\alpha$ or $\beta$ is not in $(0, 1/2)$.
we know that if $\alpha=0$, it is exactly the same with LPT, which cannot beat $1.5$; if $\beta=0$, a counter example can be given that is similar to that of Candidate Algorithm 1. 
Next we show that when $\beta \ge 1/2$ or $\alpha \ge 1/2$, \lptr could not beat $1.5$, no matter what the value of the other parameter is.

\paragraph{Candidate Algorithm 4}
In \lptr, set $\beta \ge 1/2$, and $\alpha$ to be any constant.

\noindent
\textbf{Counter example.} At $t=0$, $m (m \ge 4)$ jobs come first, with $p_1=p_2=\cdots=p_m=1$. After all these jobs are scheduled, still at time $0$, another $m$ jobs come, with $p_{m+1} = p_{m+2} = \cdots = p_{2m} = 3/2+\xi$. Then at $t=1$, another job comes with $p_{2m+1}=2$. Since $\beta\geq 1/2$, no replacement happens, and $\alg=9/2+\xi$, while in optimal schedule, all jobs could be completed at or before time $3+\xi$.

\paragraph{Candidate Algorithm 5}
In \lptr, set $\alpha \ge 1/2$, and $\beta$ to be any constant such that $0<\beta<1/2$.

\noindent
\textbf{Counter example.} We consider the special case with only one machine. At $t=0$, job $1$ comes with $p_1=1$. At $t=1-\xi$ (again, $\xi$ is an infinitesimal amount), job $2$ comes with $p_2=2$, then job $1$ is replaced by job $2$. At $t=3-2\xi$, job $3$ comes with $p_3=4$, job $2$ is replaced by job $3$. The same thing goes on and on. Each time a new job comes, $\alg$ would replace the previous job, while $\opt$ would wait for the previous job to end. The final ratio would be arbitrarily close to $1.5$.

\section{Hardness for deterministic algorithms with restart} \label{appendix:restart_hardness}

In this section, we present a simple lower bound of $\sqrt{1.5} \approx 1.225$ for any deterministic algorithms with restart.
Given any deterministic algorithm, consider the following instance with two machines.

At $t=0$, two jobs come with size $p_1 = p_2 = 1$. 
Then, at $t=3-\sqrt{6}$, another job comes with size $p_3=\sqrt{6}-1$. 

\begin{enumerate}
	\item If the algorithm starts processing job $3$ after time $1$, i.e., after completing the two jobs of size $1$, no more jobs arrive in the instance. 
	We have $\opt = 2$ as we could have scheduled job $1$ and job $2$ on the same machine and job $3$ on the other one.
	On the other hand, $\alg \ge 1 + p_3 = \sqrt{6}$.
	\item If the algorithm starts processing job $3$ before time $1$, e.g., it restarts one of the size-$1$ jobs, let there be a fourth job that arrives at time $1$ with size $p_4 = \sqrt{6} - 1$. 
	We have $\opt = \sqrt{6}$ by scheduling job $1$ and $3$ on one machine, and $2$ and $4$ on the other.
	On the other hand, we have $\alg \geq 3$ since at time $1$ at least one of jobs $1$ and $2$ is pending, and job $3$ does not complete until time $2$.
\end{enumerate}

{
	\bibliography{scheduling}
	\bibliographystyle{plainnat}
}

\end{document}